\documentclass{amsart}

\usepackage[margin=1in]{geometry}
\usepackage{graphicx} 
\usepackage{amsfonts,amsmath,amssymb}
\usepackage{float}
\newtheorem{theorem}{Theorem}[section]
\newtheorem{proposition}[theorem]{Proposition}
\newtheorem{lemma}[theorem]{Lemma}
\newtheorem{corollary}[theorem]{Corollary}
\newtheorem{definition}[theorem]{Definition}

\newcommand{\bE}{\ensuremath{\mathbf{E}}}

\begin{document}

\title[A constructive algorithm for the LLL on permutations]{A constructive algorithm for the Lov\'{a}sz Local Lemma on permutations$^{1}$}

\author[David G. Harris and Aravind Srinivasan]{
{\sc David G.~Harris}$^{2}$
\and
{\sc Aravind Srinivasan}$^{3}$
}

\setcounter{footnote}{0}
\addtocounter{footnote}{1}
\footnotetext{This is an extended version of a paper which appeared in the \emph{Proc.\ ACM-SIAM Symposium on Discrete Algorithms}, (SODA) 2014.}

\addtocounter{footnote}{1}
\footnotetext{Department of Applied Mathematics, University of Maryland, 
College Park, MD 20742. 
Research supported in part by NSF Awards CNS-1010789 and CCF-1422569.
Email: \texttt{davidgharris29@hotmail.com}.}
\addtocounter{footnote}{1}
\footnotetext{Department of Computer Science and
Institute for Advanced Computer Studies, University of Maryland, 
College Park, MD 20742. 
Research supported in part by NSF Awards CNS-1010789 and CCF-1422569, and by a research award from Adobe, Inc. 
Email: \texttt{srin@cs.umd.edu}.}

\maketitle
\begin{abstract}
 While there has been significant progress on algorithmic aspects of the Lov\'{a}sz Local Lemma (LLL) in recent years, a noteworthy exception is when the LLL is used in the context of random permutations. The breakthrough algorithm of Moser \& Tardos only works in the setting of independent variables, and does not apply in this context. 
We resolve this by developing a randomized polynomial-time algorithm for such applications. A noteworthy application is for Latin transversals: the best-known general result here (Bissacot et al., improving on Erd\H{o}s and Spencer), states that any $n \times n$ matrix  in which each entry appears at most $(27/256)n$ times, has a Latin transversal. We present the first polynomial-time algorithm to construct such a transversal. We also develop RNC algorithms for Latin transversals, rainbow Hamiltonian cycles, strong chromatic number, and hypergraph packing. 

In addition to efficiently finding a configuration which avoids bad-events, the algorithm of Moser \& Tardos has many powerful extensions and properties. These include a well-characterized distribution on the output distribution, parallel algorithms, and a partial resampling variant. We show that our algorithm has nearly all of the same useful properties as the Moser-Tardos algorithm, and present a comparison of this aspect with recent works on the LLL in general probability spaces. 
\end{abstract}



\section{Introduction}
Recent years have seen substantial progress in developing algorithmic versions of the Lov\'{a}sz Local Lemma (LLL) and some of its generalizations, starting with the breakthrough work of Moser \& Tardos \cite{moser-tardos}: see, e.g., \cite{haeupler,harris-srin2,DBLP:conf/stoc/KolipakaS11,pegden}. However, one major relative of the LLL that has eluded constructive versions, is the ``lopsided" version of the LLL (with the single exception of the CNF-SAT problem \cite{moser-tardos}).
 A natural setting for the lopsided LLL is where we have one or many random permutations \cite{erdos-spencer,lu,mohr}.  This approach has been used for Latin transversals \cite{bissacot,erdos-spencer,szabo:transversals}, hypergraph packing \cite{random-inj}, certain types of graph coloring \cite{bottcher}, and in proving the existence of certain error-correcting codes \cite{keevash}. However, current techniques do not give constructive versions in this context. We develop a randomized polynomial-time algorithm to construct such permutation(s) whose existence is guaranteed by the lopsided LLL, leading to several algorithmic applications in combinatorics. Furthermore, since the appearance of the
conference version of this work \cite{harris-srin-lllperm}, related works including \cite{achlioptas,harvey,kolmogorov} have been published; we make a comparison to these in 
Sections \ref{sec:intro-comparison} and \ref{achlioptas-sec}, detailing which of our contributions do not appear to follow from the frameworks of
\cite{achlioptas,harvey,kolmogorov}. 

\subsection{The Lopsided Lov\'{a}sz Local Lemma and random permutations}
\label{sec:perm-lll}
Suppose we want to select permutations $\pi_1, \dots, \pi_N$, where each $\pi_k$ is a permutation on the set $[n_k] = \{1, \dots, n_k\}$. In addition we have a set $\mathcal B$ of ``bad events.'' We want to select permutations $\pi$ such that no bad event is true. The \emph{lopsided version} of the Lov\'{a}sz Local Lemma (LLL) can be used to prove that such permutations exist, under suitable conditions.

We suppose that that the family of bad events $\mathcal B$ consists of \emph{atomic bad-events}. That is, each bad event $B \in \mathcal B$ is a set of tuples $B = \{(k_1, x_1, y_1), \dots, (k_r, x_r, y_r)\}$; it is true iff we have $(\pi_{k_1} (x_1) = y_1) \wedge \dots \wedge (\pi_{k_r} (x_r) = y_r)$. (Complex bad-events can usually be decomposed into atomic bad-events, so this does not lose much generality.) We will assume that no bad-event contains two tuples $(k, x,y), (k,x,y')$ where $y \neq y'$, or two tuples $(k,x,y), (k,x',y)$ where $x \neq x'$; such a bad-event would have probability zero, and could be ignored.

To apply the Lopsided Local Lemma in this setting, we need to define a \emph{dependency graph} with respect to these bad events. We connect two bad events $B, B'$ by an edge iff they overlap in one slice of the domain or range of a permutation; namely, iff there are some $k, x, y_1, y_2$ with $(k, x, y_1) \in B, (k, x, y_2) \in B'$ \emph{or} there are some $k, x_1, x_2, y$ with $(k, x_1, y) \in B, (k, x_2, y) \in B'$.  We write this $B \sim B'$; note that $B \sim B$. 
The following notation will be useful: for pairs $(x_1, y_1), (x_2, y_2)$, we write $(x_1, y_1) \sim (x_2, y_2)$ if $x_1 = x_2$ \emph{or} $y_1 = y_2$ (or both). Thus, another way to write $B \sim B'$ is that ``there are $(k,x,y) \in B$, $(k, x', y') \in B'$ with $(x,y) \sim (x', y')$".
We will use the following notation at various points: we write $(k,x,*)$ to mean any (or all) triples of the form $(k,x,y)$, and similarly for $(k,*,y)$, or $(x,*)$ etc. Therefore, yet another way to write the condition $B \sim B'$ is that there are $(k,x,*) \in B, (k,x,*) \in B'$ or $(k,*,y) \in B, (k,*,y) \in B'$.

Now suppose we select each $\pi_k$ uniformly at random and independently. This defines a probability space $\Omega$, to which we can apply the lopsided LLL. One can show that the probability of avoiding a bad event $B$ can only be \emph{increased} by avoiding other bad events $B' \not \sim B$ \cite{random-inj}. Thus, in the language of the lopsided LLL, the relation $\sim$ defines a \emph{negative-dependence} graph among the bad-events. 
(See \cite{lu,random-inj,mohr} for a study of the connection between negative dependence, random injections/permutations, and the lopsided LLL.)
Hence, the standard lopsided-LLL criterion is as follows:
\begin{theorem}[\cite{random-inj}]
\label{thm:lopsided}
Suppose that there is some assignment $\mu : \mathcal B \rightarrow \mathbf [0, \infty)$ such that for all bad-events $B \in \mathcal B$ we have
$$
\mu(B) \geq P_{\Omega}(B) \prod_{B' \sim B} (1 + \mu(B')).
$$
Then the random process of selecting each $\pi_k$ uniformly at random and independently has a positive probability of selecting permutations that avoid all the bad-events.
\end{theorem}

\smallskip \noindent \textbf{Remark:} The condition of Theorem~\ref{thm:lopsided} about the existence of such a $\mu$ is equivalent to the more-familiar LLL formulation ``there exists $x : \mathcal B \rightarrow \mathbf [0, 1)$
such that for all $B \in \mathcal B$, 
$P_{\Omega}(B) \leq x(B) \prod_{B' \sim B: B' \not= B} (1 - x(B'))$": just set $\mu(B) = x(B)/(1 - x(B))$.

\smallskip
The ``positive probability" of Theorem~\ref{thm:lopsided} is however typically exponentially small, as is standard for the LLL. 
As mentioned above, a variety of papers have used the framework of Theorem~\ref{thm:lopsided} for proving the existence of various combinatorial structures. 
Unfortunately, the algorithms for the LLL, such as Moser-Tardos resampling \cite{moser-tardos}, do not apply in this setting. The problem is that such algorithms have a more restrictive notion of when two bad-events are dependent; namely, that they share variables. (The Moser-Tardos algorithm allows for a restricted type of dependence called \emph{lopsidependence}: two bad-events which share a variable but always \emph{agree} on that value, are counted as independent. This is not strong enough to generate permutations.) So we do not have an efficient algorithm to generate such permutations, we can merely show that they exist.

We develop an algorithmic analogue of the LLL for permutations. The necessary conditions for our Swapping Algorithm are the same as for the LLL (Theorem~\ref{thm:lopsided}); however, we will construct such permutations in randomized polynomial (typically linear or near-linear) time. Our setting is far more complex than in similar contexts such as those of \cite{moser-tardos,harris-srin2,pegden}, and requires many intermediate results first. The main complication is that when we encounter a bad event involving ``$\pi_k(x) = y$", and perform our algorithm's random swap associated with it, we could potentially be changing any entry of $\pi_k$. In contrast, when we resample a variable in \cite{moser-tardos,harris-srin2,pegden}, all the changes are confined to that variable. There is a further technical issue: the current witness-tree-based algorithmic versions of the LLL
such as \cite{moser-tardos,harris-srin2}, identify, for each bad-event $B$ in the witness-tree $\tau$, some necessary event occurring with probability at most $P_{\Omega}(B)$. This is not the proof we employ here; there are significant additional terms (``$(n_k - A_{k}^{0})!/n!$'' -- see the proof of 
Lemma~\ref{witness-tree-lemma}) that are gradually ``discharged'' over time.

 We also develop RNC versions of our algorithms. Going from serial to parallel is fairly direct in \cite{moser-tardos}; our main bottleneck here is that when we resample an ``independent" set of bad events, they could still influence each other. 

(Note: we distinguish in this paper between the probability of events which occur in our algorithm, which we denote simply by $P$, and the probabilities of events within the space $\Omega$, which we denote by $P_{\Omega}$.)

\subsection{Comparison with other LLLL algorithms}
\label{sec:intro-comparison}
Building on an earlier version of this work \cite{harris-srin-lllperm}, there have been several papers which have developed generic frameworks for variations of the Moser-Tardos algorithm applied to other probability spaces. In \cite{achlioptas},  Achlioptas \& Iliopoulos gave an algorithm which is based on a compression analysis for a random walk; this was improved for permutations and matchings by Kolmogorov \cite{kolmogorov}. In \cite{harvey}, Harvey \& Vondr\'{a}k gave a probabilistic analysis similar to the parallel Moser-Tardos algorithm. These frameworks both include the permutation LLL as well as some other combinatorial applications. These papers give much simpler proofs that the Swapping Algorithm terminates quickly. 

The Moser-Tardos algorithm has many other powerful properties and extensions, beyond the fact that it efficiently finds a configuration avoiding bad-events. These properties include a well-characterized distribution on the output distribution at the end of the resampling process, a corresponding efficient parallel (RNC) algorithm, a partial-resampling variant (as developed in \cite{harris-srin2}), and an arbitrary (even adversarial) choice of which bad-event to resample. All of these properties follow from the Witness Tree Lemma we show for our Swapping Algorithm. The more generalized LLLL frameworks of \cite{achlioptas,harvey} have a limited ability to show such extensions. 

We will discuss the relationship between this paper and the other LLLL frameworks further in Section~\ref{achlioptas-sec}. As one example of the power of our proof method, we develop a parallel Swapping Algorithm in Section~\ref{sec:parallel}; we emphasize that such a parallel algorithm cannot be shown using the results of \cite{achlioptas} or \cite{harvey}. A second example is provided by Theorem~\ref{szabo-thm}, results such as which we do not see how to develop using the frameworks of \cite{achlioptas,harvey,kolmogorov}. 

One of the main goals of our paper is to provide a model for what properties a generalized LLLL algorithm should have. In our view, there has been significant progress toward this goal but there remain many missing pieces toward a \emph{true} generalization of the Moser-Tardos algorithm. We will discuss this more in a concluding section, Section~\ref{sec:conclusion}.

\subsection{Applications}
\label{sec:app-description}
We present algorithmic applications for four classical combinatorial problems: Latin transversals, rainbow Hamiltonian cycles, strong chromatic number, and edge-disjoint hypergraph packing. In addition to the improved bounds, we wish to highlight that our algorithmic approach can go beyond Theorem~\ref{thm:lopsided}: as we will see shortly, one of our (asymptotically-optimal) algorithmic results on Latin transversals, could not even have been shown non-constructively using the lopsided LLL prior to this work. 

The study of Latin squares and the closely-related Latin transversals is a classical area of combinatorics, going back to Euler and earlier \cite{denes-keedwell}. Given an $m \times n$ matrix $A$ with $m \leq n$, a \emph{transversal} of $A$ is a choice of $m$ elements from $A$, one from each row and at most one from any column. Perhaps the major open problem here is: given an integer $s$, under what conditions will $A$ have an \emph{$s$-transversal}: a transversal in which no value appears more than $s$ times \cite{bissacot,erdos-etal-latinsq,erdos-spencer,shor:latin,stein:latin,szabo:transversals}? The usual type of sufficient condition sought here is an upper bound $\Delta$ on the number of occurrences of any given value in $A$. That is, we ask: what is the maximum $\Delta$ such that any $m \times n$ matrix $A$ in which each value appears at most $\Delta$ times, is guaranteed to have an $s$-transversal? We denote this quantity by $L(s; m,n)$. The case $s = 1$ is perhaps most studied, and $1$-transversals are also called \emph{Latin transversals}. The case $m = n$ is also commonly studied (and includes Latin squares as a special case), and we will also focus on these. It is well-known that $L(1; n, n) \leq n - 1$ \cite{stein:latin}. In perhaps the first application of the lopsided LLL to random permutations, Erd\H{o}s \& Spencer essentially proved a result very similar to Theorem~\ref{thm:lopsided}, and used it to show that $L(1; n, n) \geq n/(4e)$ \cite{erdos-spencer}. (Their paper shows that $L(1; n, n) \geq n/16$; the $n/(4e)$ lower-bound follows easily from their technique.) To our knowledge, this is the first $\Omega(n)$ lower-bound on $L(1; n, n)$. Alon asked if there is a constructive version of this result \cite{alon:prob-proofs}. Building on \cite{erdos-spencer} and using the connections to the LLL from 
\cite{scott-sokal,shearer:lll}, Bissacot \emph{et al}.\ showed non-constructively that $L(1; n, n) \geq (27/256) n$ \cite{bissacot}. Our result makes these results constructive. 

The lopsided LLL has also been used to study the case $s > 1$ \cite{szabo:transversals}. Here, we prove a result that is asymptotically optimal for large $s$, except for the lower-order $O(\sqrt{s})$ term: we show (algorithmically) that $L(s; n, n) \geq (s - O(\sqrt{s})) \cdot n$. An interesting fact is that this was not known even non-constructively before: Theorem~\ref{thm:lopsided} roughly gives $L(s; n, n) \geq (s/e) \cdot n$. We also give faster serial and perhaps the first RNC algorithms with good bounds, for the strong chromatic number. Strong coloring is quite well-studied \cite{alon1992,axenovich,fellows,haxell2004,haxell2008}, and is in turn useful in \emph{covering} a matrix with Latin transversals \cite{alon-spencer-tetali}.

\subsection{Outline}
In Section~\ref{sec:swap-alg} we introduce our Swapping Algorithm, a variant of the Moser-Tardos resampling algorithm. In it, we randomly select our initial permutations; as long as some bad-event is currently true, we perform certain random swaps to randomize (or resample) them.

Section~\ref{witness-tree-sec} introduces the key analytic tools to understand the behavior of the Swapping Algorithm, namely the witness tree and the witness subdag. The construction for witness trees follows \cite{moser-tardos}; it provides an explanation or history for the random choices used in each resampling. The witness subdag is a related concept, which is new here; it provides a history not for each resampling, but for each individual swapping operation performed during the resamplings.

In Section~\ref{perm-conditions}, we show how these witness subdags may be used to deduce partial information about the permutations. As the Swapping Algorithm proceeds in time, the witness subdags can also be considered to evolve over time. At each stage of this process, the current value of the witness subdags provides information about the current values of the permutations.
In Section~\ref{total-prob-sec}, we use this process to make probabilistic predictions for certain swaps made by the Swapping Algorithm: namely, whenever the witness subdags change, the swaps must be highly constrained so that the permutations still conform to them. We calculate the probability that the swaps satisfy these constraints.

Section~\ref{constructive-lll-sec} puts the analyses of Sections~\ref{witness-tree-sec}, \ref{perm-conditions}, \ref{total-prob-sec} together, to prove that our Swapping Algorithm terminates in polynomial time under the same conditions as those of Theorem~\ref{thm:lopsided}; also, as mentioned in Section~\ref{sec:intro-comparison}, Section~\ref{achlioptas-sec} discusses certain contributions that our approach leads to that do not appear to follow from  
\cite{achlioptas,harvey,kolmogorov}.

In Section~\ref{sec:parallel}, we introduce a parallel (RNC) algorithm corresponding to the Swapping Algorithm. This is similar in spirit to the Parallel Resampling Algorithm of Moser \& Tardos. In the latter algorithm, one repeatedly selects a maximal independent set (MIS) of bad-events which are currently true, and resamples them in parallel. In our setting, bad-events which are ``independent'' in the LLL sense (that is, they are not connected via $\sim$), may still influence each other; a great deal of care must be taken to avoid these conflicts.

Section~\ref{alg-sec} describes a variety of combinatorial problems to which our Swapping Algorithm can be applied, including Latin transversals, strong chromatic number, and hypergraph packing. Finally, we conclude in Section~\ref{sec:conclusion} with a discussion of future goals for the construction of a generalized LLL algorithm.

\section{The Swapping Algorithm}
\label{sec:swap-alg}
We will analyze the following \emph{Swapping Algorithm} algorithm to find a satisfactory $\pi_1, \dots, \pi_N$:
\begin{enumerate}
\item Generate the permutations $\pi_1, \dots, \pi_N$ uniformly at random and independently. 
\item While there is some true bad-event:
\begin{enumerate}
\item[(3)] Choose some true bad-event $B \in \mathcal B$ arbitrarily. For each permutation that is involved in $B$, we perform a \emph{swapping} of all the relevant entries. (We will describe the swapping subroutine ``Swap" shortly.) We refer to this step as a \emph{resampling} of the bad-event $B$.

\smallskip 
Each permutation involved in $B$ is swapped independently, but if $B$ involves multiple entries from a single permutation, then all such entries are swapped \emph{simultaneously}. For example, if $B$ consisted of triples $(k_1, x_1, y_1), (k_2, x_2, y_2), (k_2, x_3, y_3)$, then we would perform 
$\text{Swap}(\pi_1; x_1)$ and
$\text{Swap}(\pi_2; x_2, x_3)$, where the ``Swap" procedure is given next. 
\end{enumerate}
\end{enumerate}

\smallskip \noindent
The swapping subroutine $\text{Swap}(\pi; x_1, \dots, x_r)$ for a permutation $\pi: [t] \rightarrow [t]$ as follows:

\smallskip \noindent
Repeat the following for $i = 1, \dots, r$:
\begin{itemize}
\item Select $x'_i$ uniformly at random among $[t] - \{x_1, \dots, x_{i-1} \}$. 
\item Swap entries $x_i$ and $x'_i$ of $\pi$.
\end{itemize}

Note that at every stage of this algorithm all the $\pi_k$ are permutations, and if this algorithm terminates, then the $\pi_k$ must avoid all the bad-events. So our task will be to show that the algorithm terminates in polynomial time.
We measure time in terms of a single iteration of the main loop of the Swapping Algorithm: each time we run one such iteration, we increment the time by one. 
We will use the notation $\pi_k^T$ to denote the value of permutation $\pi_k$ after time $T$. The initial sampling of the permutation (after Step (1)) generates $\pi_k^0$. 

The swapping subroutine seems strange; it would appear more natural to allow $x'_i$ to be uniformly selected among $[t]$. However, the swapping subroutine is nothing more than than the Fisher-Yates Shuffle for generating uniformly-random permutations. If we allowed $x'_i$ to be chosen from $[t]$ then the resulting permutation would be biased. The goal is to change $\pi_k$ in a minimal way to ensure that $\pi_k(x_1), \dots, \pi_k(x_r)$ and $\pi_k^{-1}(y_1), \dots, \pi_k^{-1}(y_r)$ are adequately randomized. 

There are alternative methods for generating random permutations, and many of these can replace the Swapping subroutine without changing our analysis. We discuss a variety of such equivalencies in Appendix~\ref{symmetry-sec}; these will be used in various parts of our proofs. We note that one class of algorithms that has a very different behavior is the commonly used method to generate random reals $r_i \in [0,1]$, and then form the permutation by sorting these reals. When encountering a bad-event, one would resample the affected reals $r_i$. In our setting, where the bad-events are defined in terms of specific values of the permutation, this is not a good swapping method because a single swap can drastically change the permutation. When bad-events are defined in terms of the relative \emph{rankings} of the permutation (e.g. a bad event is $\pi(x_1) < \pi(x_2) < \pi(x_3)$), then this is a better method and can be analyzed in the framework of the ordinary Moser-Tardos algorithm.

\section{Witness trees and witness subdags}
\label{witness-tree-sec}
To analyze the Swapping Algorithm, following the Moser-Tardos approach \cite{moser-tardos}, we introduce the concept of an execution log and a witness tree. The execution log consists of listing every resampled bad-event, in the order that they are resampled. We form a witness tree to justify the resampling at time $t$. We start with the resampled bad-event $B$ corresponding to time $t$, and create a single node in our tree labeled by this event. 
We move backward in time; for each bad-event $B$ we encounter, we add it to the witness tree if $B \sim B'$ for some event $B'$ already in the tree: we choose such a $B'$ that has the maximum depth in the current tree (breaking ties arbitrarily), and make $B$ a child of this $B'$ (there could be many nodes labeled $B'$). If $B \not\sim B'$ for all $B'$ in the current tree, we ignore this $B$ and keep moving backward in time. 
To make this discussion simpler we say that the root of the tree is at the ``top'' and the deep layers of the tree are at the ``bottom''. The top of the tree corresponds to later events, the bottom of the tree to the earliest events.

For the remainder of this section, the dependence on the ``justified'' bad-event at time $t$ at the root of the tree will be understood; we will omit it from the notation.

We will use the term ``witness tree'' in two closely-related senses in the following proof. First, when we run the Swapping Algorithm, we produce a witness tree $\hat \tau^T$; this is a random variable. Second, we might want to fix some labeled tree $\tau$, and discuss hypothetically under what conditions it could be produced or what properties it has; in this sense, $\tau$ is a specific object. We will always use the notation $\hat \tau^T$ to denote the specific witness tree produced by running the Swapping Algorithm, corresponding to resampling time $T$. We write $\hat \tau$ as short-hand for $\hat \tau^T$ where $T$ is understood from context (or irrelevant).

If $\tau$ is a witness tree, we say that $\tau$ \emph{appears} iff $\hat \tau^T = \tau$ for some $T \geq 0$.

The critical lemma that allows us to analyze the behavior of this algorithm is the \emph{Witness Tree Lemma}: 
\begin{lemma}[Witness Tree Lemma]
\label{witness-tree-lemma}
Let $\tau$ be a witness tree, with nodes labeled $B_1, \dots, B_s$. The probability that $\tau$ was produced as the witness tree corresponding any \emph{any} resampling time $t \geq 0$, is at most
$$
P(\text{$\tau$ appears}) \leq P_{\Omega}(B_1) \cdots P_{\Omega}(B_s)
$$

Note that the probability of the event $B$ within the space $\Omega$ can be computed as follows: if $B$ contains $r_1, \dots, r_N$ elements from each of the permutations $1, \dots, N$, (and $B$ is not impossible) then we have $$
P_{\Omega}(B) = \frac{(n_1-r_1)!}{n_1!} \dots \frac{ (n_N - r_N)! }{n_N!}
$$
\end{lemma}
This lemma is superficially similar to the corresponding lemma in Moser-Tardos \cite{moser-tardos}. However, the proof will be far more complex, and we will require many intermediate results first. The main complication is that when we encounter a bad-event involving $\pi_k(x) = y$, and we perform the random swap associated with it, then we could potentially be changing any entry of $\pi_k$. By contrast, in the usual Moser-Tardos algorithm, when we resample a variable, all the changes are confined to that variable. However, as we will see, the witness tree will leave us with enough clues about which swap was actually performed that we will be able to narrow down the possible impact of the swap.

The analysis in the next sections can be very complicated. We have two recommendations to make these proofs easier. First, the basic idea behind how to form and analyze these trees comes from \cite{moser-tardos}; the reader should consult that paper for results and examples which we omit here. Second, one can get most of the intuition behind these proofs by considering the situation in which there is a single permutation, and the bad-events all involve just a single element; that is, every bad-event has the form $\pi(x_i) = y_i$. In this case, the witness subdags (defined later) are more or less equivalent to the witness tree. (The main point of the witness subdag concept is, in effect, to reduce bad-events to their individual elements.) When reading the following proofs, it is a good idea to keep this special case in mind. In several places, we will discuss how certain results simplify in that setting.

The following proposition is the main reason the witness tree encodes sufficient information about the sequence of swaps:
\begin{proposition}
\label{change-prop}
Suppose that at some time $t_0$ we have $\pi^{t_0}_k(X) \neq Y$, and at some later time $t_2 > t_0$ we have $\pi^{t_2}_k(X) = Y$. Then there must have occurred at some intermediate time $t_1$ some bad-event including $(k, X,*)$ or $(k, *, Y)$.
\end{proposition}
\begin{proof}
Let $t_1 \in [t_0, t_2 - 1]$ denote the earliest time at which we had $\pi^{t_1+1}(X) = Y$; this must be due to encountering some bad-event including the elements $(k,x_1, y_1), \dots, (k,x_r,y_r)$ (and possibly other elements from other permutations). Suppose that $\pi_k(X) = Y$ was first caused by swapping entry $x_i$, which at that time had $\pi_k(x_i) = y'_i$, with some $x''$.

After this swap, we have $\pi_k(x_i) = y''$ and $\pi_k(x'') = y'_i$. Evidently $x'' = X$ or $x_i = X$. In the second case, the bad event at time $t_1$ included $(k, X, *)$ as desired and we are done.

So suppose $x''  = X$ and $y'_i = Y$. So at the time of the swap, we had $\pi_k(x_i) = Y$. The only earlier swaps in this resampling were with $x_1, \dots, x_{i-1}$; so at the beginning of time $t_1$, we must have had $\pi^{t_1}_k(x_j) = Y$ for some $j \leq i$. This implies that $y_j = Y$, so that the bad-event at time $t_1$ included $(k,*,Y)$ as desired.
\end{proof}

To explain some of the intuition behind Lemma~\ref{witness-tree-lemma}, we note that Proposition~\ref{change-prop} implies Lemma~\ref{witness-tree-lemma} for a \emph{singleton} witness tree.
\begin{corollary}
Suppose that $\tau$ is a singleton node labeled by $B$. Then $P(\text{$\tau$ appears}) \leq P_{\Omega}(B)$.
\end{corollary}
\begin{proof}
Suppose $\hat \tau^T = \tau$. We claim that $B$ must have been true of the initial configuration. For suppose that $(k, x,y) \in B$ but in the initial configuration we have $\pi_k(x) \neq y$. At some later point in time $t \leq T$, the event $B$ must become true. By Proposition~\ref{change-prop}, then there is some time $t' < t$ at which we encounter a bad-event $B'$ including $(k,x,*)$ or $(k,*,y)$. This bad-event $B'$ occurs earlier than $B$, and $B' \sim B$. Hence, we would have placed $B'$ below $B$ in the witness tree $\hat \tau^T$.
\end{proof}

In proving Lemma~\ref{witness-tree-lemma}, we will \emph{not} need to analyze the interactions between the separate permutations, but rather we will be able to handle each permutation in a completely independent way. For a permutation $\pi_k$, we define the \emph{witness subdag for permutation $\pi_k$}; this is a relative of the witness tree, but which only includes the information for a single permutation at a time.

\begin{definition}[witness subdags]
For a permutation $\pi_k$, a \emph{witness subdag for $\pi_k$} is defined to be a directed acyclic simple graph, whose nodes are labeled with pairs of the form $(x,y)$. If a node $v$ is labeled by $(x,y)$, we write $v \approx (x,y)$. This graph must in addition satisfy the following properties:
\begin{enumerate}
\item If any pair of nodes overlaps in a coordinate, that is, we have $v \approx (x,y) \sim (x', y') \approx v'$, then nodes $v, v'$ must be comparable (that is, either there is a path from $v$ to $v'$ or vice-versa). 
\item Every node of $G$ has in-degree at most two and out-degree at most two.
\end{enumerate}

We also may label the nodes with some auxiliary information, for example we will record that the nodes of a witness subdag correspond to bad-events or nodes in a witness tree $\tau$.

We will use the same terminology as for witness trees: vertices on the ``bottom'' are close to the source nodes of $G$ (appearing earliest in time), and vertices on the ``top'' are close to the sink nodes of $G$ (appear latest in time).
\end{definition}

The witness subdags that we will be interested in are derived from witness trees in the following manner.
\begin{definition}[Projection of a witness tree]
For a witness tree $\tau$, we define the \emph{projection of $\tau$ onto permutation $\pi_k$} which we denote $\text{Proj}_k(\tau)$, as follows.

Suppose we have a node $v \in \tau$ which is labeled by some bad-event $B = (k_1, x_1, y_1), \dots, (k_r, x_r, y_r)$. For each $i$ with $k_i = k$, we create a corresponding node $v'_i \approx (x_i, y_i)$ in the graph $\text{Proj}_k(\tau)$. We also include some auxiliary information indicating that these nodes came from bad event $B$, and in particular that all such nodes are part of the same bad-event.

The edges of $\text{Proj}_k(\tau)$ are formed follows. For each node $v' \in \text{Proj}_k(\tau)$, labeled by $(x,y)$ and corresponding to $v \in \tau$, we find the node $w_x \in \tau$ (if any) which satisfies the following properties:
\begin{enumerate}
\item[(P1)] The depth of $w_x$ is smaller than the depth of $v$
\item[(P2)] $w_x$ is labeled by some bad-event $B'$ which contains $(k,x, *)$
\item[(P3)] Among all vertices satisfying (P1), (P2), the depth of $w_x$ is maximial
\end{enumerate}

If this node $w_x \in \tau$ exists, then it corresponds to a node $w_x' \in \text{Proj}_k(\tau)$ labeled $(k, x, *)$; we construct an edge from $v'$ to $w_x'$. Note that, since the levels of the witness tree are independent under $\sim$, there can be at most one such $w_x$ and at most one such $w_x'$.

We similary define a node $w_y$ satisfying:
\begin{enumerate}
\item[(P1')] The depth of $w_y$ is smaller than the depth of $v$
\item[(P2')] $w_y$ is labeled by some bad-event $B'$ which contains $(k,y, *)$
\item[(P3')] Among all vertices satisfying (P1), (P2), the depth of $w_y$ is maximial
\end{enumerate}
If this node exists, we create an edge from $v'$ to the corresponding $w_y' \in \text{Proj}_k(\tau)$ labeled $(k,*,y)$.

Note that since edges in $\text{Proj}_k(\tau)$ correspond to \emph{strictly} smaller depth in $\tau$, the graph $\text{Proj}_k (\tau)$ is acyclic. Also, note that it is possible that $w_x = w_y$; in this case we only add a single edge to $\text{Proj}_k(\tau)$.

\textbf{Expository Remark:} In the special case when each bad-event contains a single element, the witness subdag is a ``flattening" of the tree structure. Each node in the tree corresponds to a node in the witness subdag, and each node in the witness subdag points to the next highest occurrence of the domain and range variables.
\end{definition}

Basically, the projection of $\tau$ onto $k$ tells us all of the swaps of $\pi_k$ that occur. It also gives us some of the temporal information about these swaps that would have been available from $\tau$. If there is a path from $v$ to $v'$ in $\text{Proj}_k(\tau)$, then we know that the swap corresponding to $v$ must come before the swap corresponding to $v'$. It is possible that there are a pair of nodes in $\text{Proj}_k(\tau)$ which are incomparable, yet in $\tau$ there was enough information to deduce which event came first (because the nodes would have been connected through some other permutation). So $\text{Proj}_k(\tau)$ does discard some information from $\tau$, but it turns out that we will not need this information.

To prove Lemma~\ref{witness-tree-lemma}, we will prove (almost) the following claim: Let $G$ be a witness subdag for permutation $\pi_k$; suppose the nodes of $G$ are labeled with bad-events $B_1, \dots, B_s$. Then the probability that there is some $T > 0$ such that $G = \text{Proj}_k(\hat \tau^T)$, is at most
\begin{equation}
\label{y1}
P(\text{$G = \text{Proj}_k(\hat \tau^T)$ for some $T > 0$}) \leq P_k(B_1) \cdots P_k(B_s)
\end{equation}
where, for a bad-event $B$ we define $P_k(B)$ in a similar manner to $P_{\Omega}(B)$; namely that if the bad-event $B$ contains $r_k$ elements from permutation $k$, then we define $P_k(B) = \frac{(n_k - r_k)!}{n_k!}$.

Unfortunately, proving this directly runs into technical complications regarding the order of conditioning. It is simpler to just sidestep these issues. However, the reader should bear in mind (\ref{y1}) as the \emph{informal} motivation for the analysis in Section~\ref{perm-conditions}.

\section{The conditions on a permutation $\pi_{k^*}$ over time}
\label{perm-conditions}
In Section~\ref{perm-conditions}, we will fix a value $k^*$, and we will describe conditions that $\pi_{k^*}^t$ must satisfy at various times $t$ during the execution of the Swapping Algorithm. \emph{In this section, we are only analyzing a single permutation $k^*$. To simplify notation, the dependence on $k^*$ will be hidden henceforth; we will discuss simply $\pi, \text{Proj}(\tau)$, and so forth. }

This analysis can be divided into three phases.
\begin{enumerate}
\item We define the \emph{future-subgraph} at time $t$, denoted $G_t$. This is a kind of graph which encodes necessary conditions on $\pi^t$, in order for $\tau$ to appear, that is, for $\hat \tau^T = \tau$ for some $T > 0$. Importantly, these conditions, and $G_t$ itself, are independent of the precise value of $T$. We define and describe some structural properties of these graphs.
\item We analyze how a future-subgraph $G_t$ imposes conditions on the corresponding permutation $\pi^t$, and how these conditions change over time.
\item We compute the probability that the swapping satisfies these conditions.
\end{enumerate}

We will prove (1) and (2) in Section~\ref{perm-conditions}. In Section~\ref{total-prob-sec} we will put this together to prove (3) for all the permutations.
\subsection{The future-subgraph}
Suppose we have fixed a target graph $G$, which could hypothetically have been produced as the projection of $\hat \tau^T$ onto $k^*$. We begin the execution of the Swapping Algorithm and see if, so far, it is still possible that $G = \text{Proj}_{k^*}(\hat \tau^T)$, or if $G$ has been disqualified somehow. Suppose we are at time $t$ of this process; we will show that certain swaps must have already occurred at past times $t' < t$, and certain other swaps must occur at future times $t' > t$.  

We define the \emph{future-subgraph} of $G$ at time $t$, denoted $G_t$, which tells us all the future swaps that must occur.  

\begin{definition}[The future-subgraph]
\label{future-defn}
We define the future-subgraphs $G_t$ inductively. Initially $G_0 = G$.  When we run the Swapping Algorithm, as we encounter a bad-event $(k_1, x_1, y_1), \dots, (k_r, x_r, y_r)$ at time $t$, we form $G_{t+1}$ from $G_t$ as follows:
\begin{enumerate}
\item Suppose that $k_i = k^*$, and $G_t$ contains a source node $v$ labeled $(x_i,y_i)$. Then $G_{t+1} = G_t - v$.
\item Suppose that $k_i = k^*$, and $G_t$ has a source labeled $(x_i, y'')$ where $y'' \neq y_i$ or $(x'', y_i)$ where $x'' \neq x_i$. Then, as will be shown in Proposition~\ref{future-prop1}, we can immediately conclude $G$ is impossible; we set $G_{t+1} = \bot$, and we can abort the execution of the Swapping Algorithm.
\item Otherwise, we set $G_{t+1} = G_t$.
\end{enumerate}
\end{definition}

\begin{proposition}
\label{future-prop1}
For any time $t \geq 0$, let $\hat \tau^T_{\geq t}$ denote the witness tree built for the event at time $T$, but only using the execution log from time $t$ onwards. Then if $\text{Proj}(\hat \tau^T) = G$ we also have $\text{Proj}(\hat \tau^T_{\geq t}) = G_t$. 

Note that if $G_t = \bot$, the latter condition is obviously impossible; in this case, we are asserting that whenever $G_t = \bot$, it is impossible to have $\text{Proj}(\hat \tau^T) = G$.
\end{proposition}
\begin{proof}
We omit $T$ from the notation, as usual. We prove this by induction on $t$. When $t = 0$, this is obviously true as $\hat \tau_{\geq 0} = \hat \tau$ and $G_0 = G$.

Suppose we have $\text{Proj}(\hat \tau) = G$; at time $t$ we encounter a bad-event $B = (k_1, x_1, y_1), \dots, (k_r, x_r, y_r)$. By inductive hypothesis, $\text{Proj}(\hat \tau_{\geq t}) = G_t$. 

Suppose first that $\hat \tau_{\geq t+1}$ does not contain any bad-events $B' \sim B$. Then, by our rule for building the witness tree, we have $\hat \tau_{\geq t} = \hat \tau_{\geq t+1}$. Hence we have $G_t = \text{Proj}(\hat \tau_{\geq t+1})$.  When we project this graph onto permutation $k$, there cannot be any source node labeled $(k, x, y)$ with $(x,y) \sim (x_i, y_i)$ as such node would be labeled with $B' \sim B$. Hence, according to our rules for updating $G_t$, we have $G_{t+1} = G_t$. So in this case we have $\hat \tau_{\geq t} = \hat \tau_{\geq t+1}$ and $G_t = G_{t+1}$ and $\text{Proj}(\hat \tau_{\geq t}) = G_t$; it follows that $\text{Proj}(\hat \tau_{\geq t+1}) = G_{t+1}$ as desired.

Next, suppose $\hat \tau_{\geq t+1}$ does contain $B' \sim B$. Then bad-event $B$ will be added to $\hat \tau_{\geq t}$, placed below any such $B'$. When we project $\hat \tau_{\geq t}$, then for each $i$ with $k_i = k^*$ we add a node $(x_i, y_i)$ to $\text{Proj}(\hat \tau_{\geq t})$. Each such node is necessarily a source node; if such a node $(x_i, y_i)$ had a predecessor $(x'', y'') \sim (x_i, y_i)$, then the node $(x'', y'')$ would correspond to an event $B'' \sim B$ placed below $B$. Hence we see that $\text{Proj}(\hat \tau_{\geq t})$ is obtained from $\text{Proj}(\hat \tau_{\geq t})$ by adding source nodes $(x_i,y_i)$ for each $(k^*, x_i, y_i) \in B$. 

So $\text{Proj}(\hat \tau_{\geq t}) = \text{Proj}(\hat \tau_{\geq t+1})$ plus the addition of source nodes for each $(k^*, x_i, y_i)$. By inductive hypothesis, $G_t = \text{Proj}(\hat \tau_{\geq t})$, so that $G_t = \text{Proj}(\hat \tau_{\geq t+1})$ plus source nodes for each $(k^*, x_i, y_i)$. Now our rule for updating $G_{t+1}$ from $G_t$ is to remove all such source nodes, so it is clear that $G_{t+1} = \text{Proj}(\hat \tau_{\geq t+1})$, as desired.

Note that in this proof, we assumed that $\text{Proj}(\hat \tau) = G$, and we never encountered the case in which $G_{t+1} = \bot$. This confirms our claim that whenever $G_{t+1} = \bot$ it is impossible to have $\text{Proj}(\hat \tau) = G$.
\end{proof}

By Proposition~\ref{future-prop1}, the witness subdag $G$ and the future-subgraphs $G_t$ have a similar shape; they are all produced by projecting witness trees of (possibly truncated) execution logs. Note that if $G = \text{Proj}(\tau)$ for some tree $\tau$, then for any bad-event $B \in \tau$, either $B$ is not represented in $G$, or all the pairs of the form $(k^*, x, y) \in B$ are represented in $G$ and are incomparable there.

The following structural decomposition of a witness subdag $G$ will be critical.
\begin{definition}[Alternating paths] Given a witness subdag $G$, we define an \emph{alternating path} in $G$ to be a simple path which alternately proceeds forward and backward along the directed edges of $G$. For a vertex $v \in G$, the \emph{forward (respectively backward) path} of $v$ in $G$, is the maximal alternating path which includes $v$ and all the forward (respectively backward) edges emanating from $v$. Because $G$ has in-degree and out-degree at most two, every vertex $v$ has a unique forward and backward path (up to reflection); this justifies our reference to ``the" forward and backward path. These paths may be even-length cycles.

Note that if $v$ is a source node, then its backward path contains just $v$ itself. This is an important type of alternating path which should always be taken into account in our definitions.
\end{definition}

One type of alternating path, which is referred to as the \emph{W-configuration}, plays a particularly important role.

\begin{definition}[The W-configuration]
Suppose $v \approx (x, y)$ has in-degree at most one, and the backward path contains an \emph{even} number of edges, terminating at vertex $v' \approx (x', y')$. We refer to this alternating path as a \emph{W-configuration}. (See Figure~\ref{fig1}.)

Any W-configuration can be written (in one of its two orientations) as a path of vertices labeled $$
(x_0, y_1), (x_1, y_1), (x_1, y_2), \dots, (x_s, y_s), (x_s, y_{s+1});
$$ here the vertices $(x_1, y_1), \dots, (x_s, y_s)$ are at the ``base'' of the W-configuration. Note here that we have written the path so that the $x$-coordinate changes, then the $y$-coordinate, then $x$, and so on. When written this way, we refer to $(x_0, y_{s+1})$ as the \emph{endpoints} of the W-configuration.

If $v \approx (x,y)$ is a source node, then it defines a W-configuration with endpoints $(x,y)$. This should not be considered a triviality or degeneracy, rather it will be the most important type of W-configuration.
\end{definition}

\begin{figure}[H]
\begin{center}
\includegraphics[trim = 0.5cm 21.5cm 6.5cm 5cm,scale=0.5,angle = 0]{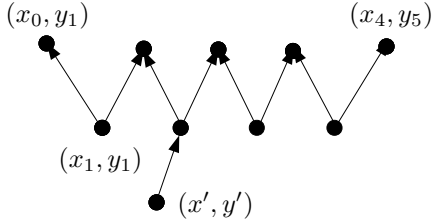}
\put(-200,70){$(x_0,y_1)$}
\put(-180,15){$(x_1,y_1)$}
\put(-135,-2){$(x', y')$}
\put(-70,70){$(x_4,y_5)$}
\caption{The vertices labeled $(x_0,y_1), (x_1, y_1), \dots, (x_4, y_5)$ form a W-configuration of length 9 with endpoints $(x_0, y_5)$. Note that the vertex $(x',y')$ is \emph{not} part of this W-configuration.
\label{fig1}}
\end{center}
\end{figure}

\subsection{The conditions on $\pi^t_{k^*}$ encoded by $G_t$}

At any $t$, the future-subgraph $G_t$ gives certain necessary conditions on $\pi$ in order for some putative $\tau$ to appear. Proposition~\ref{future-prop2} describes a certain set of conditions that plays a key role in the analysis.

\begin{proposition}
\label{future-prop2}
For any graph $G$ and integers $t \leq T$, the following condition is necessary to have $G = \text{Proj}(\hat \tau_{\geq t}^T)$:

\emph{For every W-configuration in $G_t$ with endpoints $(x_0, y_{s+1})$, we must have $\pi^t(x_0) = y_{s+1}$},
where $\pi^t$ denotes the value of the permutation at time $t$.

For example, if $v \approx (x,y)$ is a source node of $G_t$, then $\pi^t(x) = y$.
\end{proposition}
\begin{proof}
We prove this by induction on $s$. The base case is $s = 0$; in this case we have a source node $(x, y)$. Suppose $\pi^t(x) \neq y$. In order for $\hat \tau^T$ to contain some bad-event containing $(k^*, x,y)$, we must at some point $t' > t$ have $\pi^{t'} (x) = y$; let $t'$ be the minimal such time. By Proposition~\ref{change-prop}, we must encounter a bad-event containing $(k^*, x, *)$ or $(k^*, *, y)$ at some intervening time $t'' < t'$. If this bad-event contains $(k^*,x,y)$ then necessarily $\pi^{t''}(x) = y$ contradicting minimality of $t'$. So there is a bad-event $(k^*,x, \neq y)$ or $(k^*, \neq x, y)$ earlier than the earliest occurrence of $\pi(x) = y$. This event $(k^*, x, \neq y)$ or $(k^*, \neq x, y)$ projects to a source node $(x, \neq y)$ or $(\neq x, y)$ in $G_t$. But then $(x,y)$ cannot also be a source node of $G_t$. 

We now prove the induction step. Suppose we have a W-configuration with base $(x_1, y_1), \dots, (x_s, y_s)$, and suppose the endpoints of this W-configuration are vertices $v, v'$ labeled $(x_0, y_1$ and $(x_s, y_{s+1})$ respectively.

At some future time $t' \geq t$ we must encounter a bad-event $B$ involving some subset of the source nodes, say that $B$ includes $(x_{i_1}, y_{i_1}), \dots, (x_{i_r}, y_{i_r})$ for $1 \leq r \leq s$. As these were necessarily source nodes, we had $\pi^{t'} (x_{i_1}) =  y_{i_1}, \dots, \pi^{t'}(x_{i_r}) = y_{i_r}$. After the swaps, these source nodes are removed and so the updated $G_{t'+1}$ has $r+1$ new W-configurations, whose length is all smaller than $s$. By inductive hypothesis, the updated permutation $\pi^{t'+1}$ must then satisfy $\pi^{t'+1}(x_0) = y_{i_1}, \pi^{t'+1}(x_{i_1}) = y_{i_2}, \dots, \pi^{t'+1}(x_{i_r}) = y_{s+1}$.

By Proposition~\ref{swap-invariant-prop2}, we may suppose without loss of generality that the resampling of the bad event first swaps $x_{i_1}, \dots, x_{i_r}$ in that order. Let $\pi'$ denote the result of these swaps; there may be additional swaps to other elements of the permutation, but we must have $\pi^{t'+1}(x_{i_l}) = \pi'(x_{i_l})$ for $l = 1, \dots, r$.

In this case, we see that evidently $x_{i_1}$ swapped with $x_{i_2}$, then $x_{i_2}$ swapped with $x_{i_3}$, and so on, until eventually $x_{i_r}$ was swapped with $x'' = (\pi^{t'})^{-1} y_{s+1}$. At this point, we have $\pi' ( x''  ) = y_{i_1}$. Later swaps during time $t'$ may swap $x''$ with some other $x$, where $(x,y) \in B$. Thus, at time $t'+1$ we either have $\pi^{t'+1}(x'') = y_{i_1}$ or $\pi^{t'+1}(x) = y_{i_1}$ where $(x,y) \in B$. Recall that $\pi^{t'+1}(x_0) = y_{i-1}$; thus either $x'' = x_0$ or $x = x_0$. 

In the latter case, $(x_0, y) \in B$. Thus implies that, when we encounter the bad-event $B$ at time $t'$, there is a source node labeled $(x_0, y) \in G_{t'}$. This node $(x_0, y)$ would also be a node in the graph $G_t$; thus $v$ has two in-neighbors in $G_t$ labeled $(x_0, y)$ and $(x_1, y_1)$, which contradicts that it is part of a W-configuration of $G_t$.

Thus, we conclude that $x'' = x_0$. This implies that we must have $(\pi^{t'})^{-1} y_s = x'' = x_0$; that is, that $\pi^{t'}(x_0) = y_s$. This in turn implies that  $\pi^t(x_0) = y_{s+1}$. For, by Proposition~\ref{change-prop}, otherwise we would have encountered a bad-event involving $(x_0, *)$ or $(*, y_{s+1})$; these would imply an additional in-neighbor of either $v$ or $v'$ respectively, which contradicts that it is part of a W-configuration of $G_t$.
\end{proof}

Proposition~\ref{future-prop2} can be viewed equally as a definition:
\begin{definition}[Active conditions of a future-subgraph]
We refer to the conditions implied by Proposition~\ref{future-prop2} as the \emph{active conditions} of the graph $G_t$. More formally, we define
$$
\text{Active}(G) = \{ (x, y)  \mid \text{ $(x,y)$ are the end-points of a $W$-configuration of $G$} \}
$$
We also define $A^{t}_k$ to be the cardinality of $\text{Active}(G_t)$, that is, the number of active conditions of permutation $\pi_k$ at time $t$. (The subscript $k$ may be omitted in context, as usual.)
\end{definition}

When we remove source nodes $(x_1, y_1),  \dots, (x_r, y_r)$ from $G_t$, the new active conditions of $G_{t+1}$ are related to $(x_1, y_1), \dots, (x_r, y_r)$ in a particular way.
\begin{lemma}
\label{active-change-lemma}
Suppose $G$ is a future-subgraph with source nodes $v_1 \approx (x_1, y_1), \dots, v_r \approx (x_r, y_r)$. Let $H = G - v_1 - \dots - v_r$ denote the graph obtained from $G$ by removing these source nodes. Then there is a set $Z \subseteq \{ (x_1, y_1), \dots, (x_r, y_r) \}$ with the following properties:
\begin{enumerate}
\item There is an \emph{injective} function $f : Z \rightarrow \text{Active}(H)$, with the property that $(x,y) \sim f( (x,y) )$ for all $(x,y) \in Z$
\item $|\text{Active}(H)| = |\text{Active}(G)| - (r - |Z|)$
\end{enumerate}

\textbf{Expository remark:} We have recommended bearing in mind the special case when each bad-event consists of a single element. In this case, we would have $r = 1$; and the stated theorem would be that either $|\text{Active}(H)| = |\text{Active}(G)| - 1$; OR we have  $|\text{Active}(H)| = |\text{Active}(G)| $ and $(x_1, y_1) \sim (x'_1, y'_1) \in \text{Active}(H)$.

Intuitively, we are saying that every node $(x, y)$ we are removing is either explicitly constrained in an ``independent way" by some new condition in the graph $H$ (corresponding to $Z$), or it is almost totally unconstrained. We will never have the bad situation in which a node $(x,y)$ is constrained, but in some implicit way depending on the previous swaps.

\end{lemma}
\begin{proof}
Let $H_i$ denote the graph $G - v_1 - \dots - v_i$.
We will recursively build up set $Z^i$ and functions $f^i: Z^i \rightarrow \text{Active}(H)$, where $Z_i \subseteq \{(x_1, y_1), \dots, (x_i, y_i)\}$, and which satisfy the given conditions up to stage $i$.

Now, suppose we remove the source node $v_i$ from $H_{i-1}$. Observe that $(x_i, y_i) \in \text{Active}(H_{i-1})$, but (unless there is some other vertex with the same label in $G$), $(x_i, y_i) \not \in \text{Active}(H_i)$. Thus, the most obvious change when we remove $v_i$ is that we destroy the active condition $(x_i, y_i)$. This may add or subtract other active conditions as well.  

We will need to update $Z^{i-1}, f^{i-1}$. Most importantly, $f^{i-1}$ may have mapped $(x_j, y_j)$ for $j < i$, to an active condition of $H_{i-1}$ which is destroyed when $v_i$ is removed. In this case, we must re-map this to a new active condition.    Note that we cannot have $f^{i-1}(x_j, y_j) = (x_i, y_i)$ for $j < i$, as $x_i \neq x_j$ and $y_i \neq y_j$.

There are now a variety of cases depending on the forward-path of $v_i$ in $H_{i-1}$. 

\begin{enumerate}
\item This forward path consists of a cycle, or the forward path terminates on both sides in forward-edges. This is the easiest case. Then no more active conditions of $H_{i-1}$ are created or destroyed. We update $Z^i = Z^{i-1}, f^i = f^{i-1}$. One active condition is removed, in net, from $H_{i-1}$; hence $|\text{Active}(H_i)| = |\text{Active}(H_{i-1}) | - 1$.

\item This forward path contains a forward edge on one side and a backward edge on the other. For example, suppose the path has the form $(X_1, Y_1), (X_1, Y_2), (X_2, Y_2), \dots, (X_s, Y_{s+1})$, where the vertices $(X_1, Y_1), \dots, (X_s, Y_s)$ are at the base, and the node $(X_1, Y_1)$ has out-degree 1, and the node $(X_s, Y_{s+1})$ has in-degree 1. Suppose that $(x_i, y_i) = (X_j, Y_j)$ for some $j \in \{1, \dots, s \}$. (See Figure~\ref{fig2}.) In this case, we do not destroy any W-configurations, but we create a new W-configuration with endpoints $(X_j, Y_{s+1}) = (x_i, Y_{s+1})$.

We now update $Z^i = Z^{i-1} \cup \{(x_i, y_i) \}$. We define $f^i = f^{i-1}$ plus we map $(x_i, y_i)$ to the new active condition $(x_i, Y_{s+1})$. In net, no active conditions were added or removed, and $|\text{Active}(H_i)| = |\text{Active}(H_{i-1})|$.
\begin{figure}[H]
\begin{center}
\includegraphics[trim = -3.5cm 21.5cm 6.5cm 5cm,scale=0.5,angle = 0]{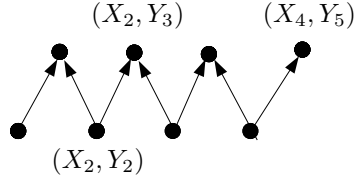}
\put(-185,15){$(X_2,Y_2)$}
\put(-170,70){$(X_2,Y_3)$}
\put(-105,70){$(X_4,Y_5)$}
\caption{When we remove $(X_2, Y_2)$, we create a new W-configuration with endpoints $(X_2, Y_5)$.
\label{fig2}}
\end{center}
\end{figure}

\item This forward path was a W-configuration $(X_0, Y_1), (X_1, Y_1), \dots, (X_s, Y_s), (X_s, Y_{s+1})$ with the pairs $(X_1, Y_1), \dots, (X_s, Y_s)$ on the base, and we had $(x_i, y_i) = (X_j, Y_j)$. This is the most complicated situation; in this case, we destroy the original W-configuration with endpoints $(X_0, Y_{s+1})$ but create two new W-configurations with endpoints $(X_0, Y_j)$ and $(X_j, Y_{s+1})$. We update $Z^i = Z^{i-1} \cup \{(x_i, y_i) \}$. We will set $f^i = f^{i-1}$, except for a few small changes as follows.

Now, suppose $f^{i-1} (x_l, y_l) = (X_0, Y_{s+1})$ for some $l < i$; so either $x_l = X_0$ or $y_l = Y_{s+1}$. If it is the former, we set $f^i(x_l, y_l) = (X_0, Y_j), f^i(x_i, y_i) = (X_j, Y_{s+1})$. If it is the latter, we set $f^i(x_l, y_l) = (X_j, Y_{s+1}), f^i(x_i, y_i) = (X_0, Y_j).$. If $(f^{i-1})^{-1} (X_0, Y_{s+1}) = \emptyset$ then we simply set $f^i(x_i, y_i) = (X_0, Y_j)$.

In any case, $f^i$ is updated appropriately, and in the net no active conditions are added or removed, so we have $|\text{Active}(H_i)| = |\text{Active}(H_{i-1})|$.
\end{enumerate}
\end{proof}

\section{The probability that the swaps are all successful}
\label{total-prob-sec}
In the previous sections, we determined necessary conditions for the permutations $\pi^t$, depending on the graphs $G_t$. In this section, we finish by computing the probability that the swapping subroutine causes the permutations to, in fact, satisfy all such conditions.

Proposition~\ref{exchange-prop2a} states the key randomness condition satisfied by the swapping subroutine.  The basic intuition behind this is as follows: suppose $\pi: [n] \rightarrow [n]$ is a fixed permutation with $\pi(x) = y$, and we call $\pi' = \text{Swap}(\pi;x_1, \dots, x_r)$. Then $\pi'(x_1)$ has a uniform distribution over $[n]$. Similarly, $\pi'^{-1}(y_1)$ has a uniform distribution over $[n]$. However, the joint distribution is \emph{not} uniform --- there is essentially only one degree of freedom for the two values. In general, any subset of the variables $\pi'(x_1), \dots, \pi'(x_r), \pi'^{-1}(y_1), \dots, \pi^{-1}(y_r)$ will have the uniform distribution, \emph{as long as the subset does not simultaneously contain $\pi'(x_i), \pi'^{-1} (y_i)$ for some $i \in [r]$.}

\begin{proposition}
\label{exchange-prop2a}
Suppose $n,r,s,q$ are non-negative integers obeying the following constraints:
\begin{enumerate}
\item $0 \leq s \leq \min(q,r)$
\item $q + (r-s) \leq n$
\end{enumerate}

Let $\pi$ be a fixed permutation of $[n]$, and let $x_1, \dots, x_r \in [n]$ be distinct, and let $y_i = \pi(x_i)$ for $i = 1, \dots, r$.  Let $(x'_1, y'_1), \dots, (x'_q, y'_q)$ be a given list with the following properties:

\begin{enumerate}
\item[(3)] All $x'$ are distinct; all $y'$ are distinct
\item[(4)] For $i = 1, \dots, s$ we have $x_i = x'_i$ or $y_i = y'_i$.
\end{enumerate}

Let $\pi' = \text{Swap}(\pi; x_1, \dots, x_r)$. Then the probability that $\pi'$ satisfies all the constraints $(x', y')$ is at most
$$
P(\pi'(x'_1) = y'_1 \wedge \dots \wedge \pi'(x'_q) = y'_q) \leq \frac{(n-r)! (n-q)!}{n! (n-q - r + s)!} 
$$

\textbf{Expository remark:} Consider the special case when each bad-event contains a single element. In that case, we have $r = 1$. There are two possibilities for $s$; either $s = 0$ in which case this probability on the right is $1 - q/n$ (i.e. the probability that $\pi'(x_1) \neq y'_1, \dots, y'_q$); or $s = 1$ in which case this probability is $1/n$ (i.e. the probability that $\pi'(x_1) = y'_1$).

\end{proposition}
\begin{proof}

Define the function $g(n,r,s,q) =  \frac{(n-r)! (n-q)!}{n! (n-q-r+s)!}$.  We will prove this proposition by induction on $s, r$. There are a few cases we handle separately:
\begin{enumerate}
\item Suppose $s > 0$ and $x_1 = x'_1$. Then, in order to satisfy the desired conditions, we must swap $x_1$ to $x'' = \pi^{-1} (y'_1)$; this occurs with probability $1/n$. The subsequent $r-1$ swaps starting with the permutation $\pi (x_1 \ x'')$ must now satisfy the conditions $\pi'(x'_2) = y'_2, \dots, \pi'(x_q) = y'_q$. We claim that we have $(x_i, \pi(x_1 \ x'') x_i) \sim (x'_i, y'_i)$ for $i = 2, \dots, s$. If $x'' \neq x_2, \dots, x_s$, this is immediately clear. Otherwise, suppose $x'' = x_j$. If $x_j = x'_j$, then we again still have $(x_j, \pi(x_1 \ x'') x_j) \sim (x'_j, y'_j)$. If $y_j = y'_j$, then this implies that $y'_1 = y_j = y'_j$, which contradicts that the $y'_j \neq y'_1$ .

 So we apply the induction hypothesis to $\pi(x_1 \ x'')$; in the induction, we subtract one from $n, q, r, s$. This gives
\begin{align*}
P(\pi'(x'_1) = y'_1 \wedge \dots \wedge \pi(x'_q) = y'_q) \leq \tfrac{1}{n} g(n-1,r-1,s-1,q-1) = g(n,r,s,q)
\end{align*}
as desired.
\item Similarly, suppose $s > 0$ and suppose $y_1 = y'_1$.  By Proposition~\ref{exchange-sym-prop}, we would obtain the same distribution if we executed $(\pi')^{-1} = \text{Swap} (\pi^{-1}; y_1, \dots, y_r)$. Hence we have
\begin{align*}
P(\pi'(x'_1) = y'_1 \wedge \dots \wedge \pi(x'_q) = y'_q) = P((\pi')^{-1}(y'_1) = x'_1 \wedge \dots \wedge (\pi')^{-1} (y'_q) = x'_q)
\end{align*}
Now, the right-hand side has swapped the roles of $x_1/y_1$; in particular, it now falls under the previous case (1) already proved, and so the right-hand side is at most $g(n,r,s,q)$ as desired.

\item Suppose $s = 0$ and that there is some $i \in [r], j \in [q]$ with $(x_i, y_i) \sim (x'_j, y'_j)$. By Proposition~\ref{swap-invariant-prop2}, we can assume without loss of generality that $(x_1, y_1) \sim (x'_1, y'_1)$. So, in this case, we are really in the case with $s = 1$. This is covered by case (1) or case (2), which have already shown. Thus, we have that
$$
P(\pi'(x'_1) = y'_1 \wedge \dots \wedge \pi(x'_q) = y'_q) \leq g(n,r,1,q) = \frac{g(n,r,0,q)}{n - q - r + 1} \leq g(n,r,s,q)
$$
Here, we are using our hypothesis that $n \geq q + (r-s) = q + r$.

\item Finally, suppose $s = 0$ and $x_1, \dots, x_r$ are distinct from $x'_1, \dots, x'_q$ and $y_1, \dots, y_q$ are distinct from $y'_1, \dots, y'_q$.  In this case, a necessary (although not sufficient) condition to have $\pi'(x'_1) = y'_1,  \dots,  \pi(x'_q) = y'_q$ is that there are some $y''_1, \dots, y''_r$, distinct from each other and distinct from $y'_1, \dots,  y'_q$, with the property that $\pi'(x_i) = y''_i$ for $j = 1, \dots, r$. By the union bound, we have
\begin{align*}
P(\pi'(x'_1) = y'_1 \wedge \dots \wedge \pi(x'_q) = y'_q)
& \leq \sum_{y''_{1}, \dots, y''_r} P(\pi'(x_1) = y''_1 \wedge \dots \wedge \pi(x_r) = y''_r)
\end{align*}
For each individual summand, we apply the induction hypothesis; the summand has probability at most $g(n,r,r,q)$. As there are $(n-q)!/(n-q-r)!$ possible values for $y''_1, \dots, y''_r$, the total probability is at most $(n-q)!/(n-q-r)! \times g(n,r,r,q) = g(n,r,s,q)$.
\end{enumerate}
\end{proof}

We apply Proposition~\ref{exchange-prop2a} to upper-bound the probability that the Swapping Algorithm successfully swaps when it encounters a bad event.
\begin{proposition}
\label{exchange-prop2}
Suppose we encounter a bad-event $B$ at time $t$ containing elements $(k, x_1, y_1)$, $\dots$, $(k, x_r, y_r)$ from permutation $k$ (and perhaps other elements from other permutations).
Then the probability that $\pi^{t+1}_k$ satisfies all the active conditions of its future-subgraph, conditional on all past events and all other swappings at time $t$, is at most 
$$
P(\text{$\pi_k^{t+1}$ satisfies $\text{Active}(G_{k}^{t+1})$}) \leq P_k(B) \frac{ (n_k - A_k^{t+1})!}{(n_k - A_k^t)!}.
$$

Recall that we have defined $A_{k}^t$ to be the number of active conditions in the future-subgraph corresponding to permutation $\pi_k$ at time $t$, and we have defined $P_k(B) = \frac{(n_k - r)!}{n_k!}.$

\textbf{Expository remark:} Consider the special case when each bad-event consists of a single element. In this case, we would have $P_k(B) = 1/n$. The stated theorem is now: either $A^{t+1} = A^t$, in which case the probability that $\pi$ satisfies its swapping condition is $1/n$; or $A^{t+1} = A^t - 1$; in which case the probability that $\pi$ satisfies its swapping condition is $1 - A^{t+1}/n$.
\end{proposition}
\begin{proof}
Let $H$ denote the future-subgraph $G_{k, t+1}$ after removing the source nodes corresponding to the pairs $(x_1, y_1), \dots, (x_{r}, y_{r})$. Using the notation of Lemma~\ref{active-change-lemma}, we set $s = |Z|$ and $q = A_k^{t+1}$. We have $\text{Active}(H) = \{ (x'_1, y'_1), \dots, (x'_{q}, y'_{q}) \}$.

For each $(x,y) \in Z$, we have $y = \pi^t(x)$, and there is an injective function $f: Z \rightarrow \text{Active}(H)$ and $(x,y) \sim f((x,y))$. By Proposition~\ref{swap-invariant-prop2}, we can assume without loss of generality $Z = \{(x_1, y_1), \dots, (x_s, y_s) \}$ and $f(x_i, y_i) = (x'_i,  y'_i)$. In order to satisfy the active conditions on $G_{k, t+1}$, the swapping must cause $\pi^{t+1} ( x'_i ) = y'_i$ for $i = 1, \dots, q$.

By Lemma~\ref{active-change-lemma}, we have $A_{k}^t = A_{k}^{t+1} + (r - s) = q + (r - s)$. Note that $A_k^t \leq n$. So all the conditions of Proposition~\ref{exchange-prop2a} are satisfied. Thus this probability is at most $\frac{(n_k-r)!}{n_k!} \times \frac{(n_k-q)!}{ (n_k-q-r+s)!} = \frac{ (n_k-r)! (n_k - A_{k}^{t+1})!}{n_k! (n_k - A_k^t)!}$.
\end{proof}
We have finally all the pieces necessary to prove Lemma~\ref{witness-tree-lemma}.

{
\renewcommand{\thetheorem}{\ref{witness-tree-lemma}}
\begin{lemma}
Let $\tau$ be a witness tree, with nodes labeled $B_1, \dots, B_s$. The probability that $\tau$ appears is at most
$$
P(\text{$\tau$  appears}) \leq P_{\Omega}(B_1) \cdots P_{\Omega}(B_s)
$$
\end{lemma}
\addtocounter{theorem}{-1}
}
\begin{proof}
The Swapping Algorithm, as we have defined it, begins by selecting the permutations uniformly at random. One may also consider fixing the permutations to some arbitrary (not random) value, and allowing the Swapping Algorithm to execute from that point onward. We refer to this as \emph{starting at an arbitrary state of the Swapping Algorithm.} We will prove the following by induction on $\tau'$: The probability, starting at an arbitrary state of the Swapping Algorithm, that the subsequent swaps would produce the subtree $\tau'$, is at most 
\begin{equation}
\label{wt1}
P(\text{$\hat \tau^T = \tau'$ for some $T \geq 0$}) \leq \prod_{B \in \tau'} P_{\Omega}(B) \times \prod_{k=1}^N \frac{n_k!}{(n_k - |\text{Active}( \text{Proj}_{k}(\tau'))|)!}. 
\end{equation}

When $\tau' = \emptyset$, the RHS of (\ref{wt1}) is equal to one so this is vacuously true.

To show the induction step, note that in order for $\tau'$ to be produced as the witness tree for some $T \geq 0$, it must be that some $B$ is resampled, where some node $v \in \tau'$ is labeled by $B$. Suppose we condition on that $v$ is the first such node, resampled at time $t$. A necessary condition to have $\hat \tau^T = \tau'$ for some $T \geq t$ is that $\pi^{t+1}$ satisfies all the active conditions on $G_{t+1}$. By Proposition~\ref{exchange-prop2}, the probability that $\pi^{t+1}$ satisfies these conditions is at most  $\prod_k P_k(B) \frac{ (n_k - A_k^{t+1})!}{(n_k - A_k^t)!}$.

Next, if this event occurs, then subsequent resamplings must cause $\hat \tau_{\geq t+1}^{T} =  \tau' - v$. To bound the probability of this, we use the induction hypothesis. Note that the induction hypothesis gives a bound conditional on \emph{any} starting configuration of the Swapping Algorithm, so we may multiply these probabilities. Thus
\begin{align*}
P(\text{$\hat \tau^T = \tau'$ for some $T > 0$}) &\leq \prod_k P_k(B) \frac{ (n_k - A_k^{t+1})!}{(n_k - A_k^t)!} \times  \prod_{B \in \tau' - v} P_{\Omega}(B) \times \prod_{k=1}^N \frac{n_k!}{(n_k - |\text{Active}( \text{Proj}_{k}(\tau' - v))|)!} \\
&= \prod_{B \in \tau'} P_{\Omega}(B) \prod_k \frac{ (n_k - A_k^{t+1})!}{(n_k - A_k^t)!} \frac{n_k!}{(n_k - |\text{Active}( \text{Proj}_{k}(\tau' - v))|)!} \\
&= \prod_{B \in \tau'} P_{\Omega}(B) \prod_k \frac{n_k!}{(n_k - A_k^t)!} \quad \text{as $A_k^{t+1} = |\text{Active}( \text{Proj}_{k}(\tau' - v))|$} 
\end{align*}
completing the induction argument.

We now consider the necessary conditions to produce the \emph{entire} witness tree $\tau$, and not just fragments of it. First, the \emph{original permutations} $\pi^0_k$ must satisfy the active conditions of the respective witness subdags $\text{Proj}_k(\tau)$. For each permutation $k$, this occurs with probability $\frac{(n_k - A_{k}^0)!}{n_k!}$. Next, the subsequent sampling must be compatible with $\tau$; by (\ref{wt1}) this has probability at most $\prod_{B \in \tau} P_{\Omega}(B) \times \prod_{k=1}^N \frac{n_k!}{(n_k - A_{k}^0)!}$. Again, note that the bound in (\ref{wt1}) is conditional on any starting position of the Swapping Algorithm, hence we may multiply these probabilities. In total we have
\begin{align*}
P(\text{$\hat \tau^T = \tau$ for some $T \geq 0$}) &\leq \prod_k \frac{(n_k - A_{k}^{0})!}{n_k!} \times \prod_{B \in \tau} P_{\Omega}(B) \times \prod_{k=1}^N \frac{n_k!}{(n_k - A_{k}^{0})!} = \prod_{B \in \tau} P_{\Omega}(B).
\end{align*}

We note one counter-intuitive aspect to this proof. The natural way of proving this lemma would be to identify, for each bad-event $B \in \tau$, some necessary event occurring with probability at most $P_{\Omega}(B)$. This is the general strategy in Moser-Tardos \cite{moser-tardos} and related constructive LLL variants such as \cite{harris-srin2}, \cite{achlioptas}, \cite{harvey}. This is \emph{not} the proof we employ here; there is an additional factor of $(n_k - A_{k}^{0})!/n!$ which is present for the original permutation and is gradually ``discharged'' as active conditions disappear from the future-subgraphs.
\end{proof}

\section{The constructive LLL for permutations}
\label{constructive-lll-sec}
Now that we have proved the Witness Tree Lemma, the remainder of the analysis is essentially the same as for the Moser-Tardos algorithm \cite{moser-tardos}. Using arguments and proofs from \cite{moser-tardos} with our key lemma, we can now easily show our key theorem:

\begin{theorem}
\label{thm:constr-lll}
Suppose there is some assignment of weights $\mu: \mathcal B \rightarrow [0, \infty)$ which satisfies, for every $B \in \mathcal B$ the condition
$$
\mu(B) \geq P_{\Omega}(B) \prod_{B' \sim B} (1 + \mu(B'))
$$
Then the Swapping Algorithm terminates with probability one. The expected number of iterations in which we resample $B$ is at most $\mu(B)$.
\end{theorem}

In the ``symmetric'' case, this gives us the well-known LLL criterion:
\begin{corollary}
Suppose each bad-event $B \in \mathcal B$ has probability at most $p$, and is dependent with at most $d$ bad-events. Then if
$e p (d+1) \leq 1$,
the Swapping Algorithm terminates with probability one; the expected number of resamplings of each bad-event is $O(1)$.
\end{corollary}

Some extensions of the LLL, such as the Moser-Tardos distribution bounds shown in \cite{haeupler}, the observation of Pegden regarding independent sets in the dependency graph \cite{pegden}, or the partial-resampling of \cite{harris-srin2}, follow almost immediately here. There are a few extensions which require slightly more discussion:

\subsection{Lopsidependence}
\label{sec:lopsi}
As in \cite{moser-tardos}, it is possible to slightly restrict the notion of dependence. Two bad-events which share the same valuation of a variable are not forced to be dependent. We can re-define the relation $\sim$ on bad-events as follows: for $B, B' \in \mathcal B$, we have $B \sim B'$ iff
\begin{enumerate}
\item $B = B'$, \emph{or}
\item there is some $(k, x, y) \in B, (k, x', y') \in B'$ with either $x = x', y \neq y'$ or $x \neq x', y = y'$.
\end{enumerate}
In particular, bad-events which share the same triple $(k,x,y)$, are \emph{not} caused to be dependent.

Proving that the Swapping Algorithm still works in this setting requires only a slight change in our definition of $\text{Proj}_k(\tau)$. Now, the tree $\tau$ may have multiple copies of any given triple $(k,x,y)$ on a single level. When this occurs, we create the corresponding nodes $v \approx (x,y) \in \text{Proj}_k(\tau)$; edges are added between such nodes in an arbitrary (but consistent) way. The remainder of the proof remains as before.

\subsection{LLL for injective functions}
\label{sec:injective}
The analysis of \cite{random-inj} considers a slightly more general setting for the LLL, in which we select random \emph{injections} $f_k: [m_k] \rightarrow [n_k]$, where $m_k \leq n_k$. In fact, our Swapping Algorithm can be extended to this case. We simply define a permutation $\pi_k$ on $[n_k]$, where the entries $\pi_k(m_k+1), \dots, \pi_k(n_k)$ are ``dummies'' which do not participate in any bad-events. The LLL criterion for the extended permutation $\pi_k$ is exactly the same as the corresponding LLL criterion for the injection $f_k$. Because all of the dummy entries have the same behavior, it is not necessary for the Swapping Algorithm to keep track of the dummy entries exactly; they are needed only for the analysis.

\subsection{Comparison with the approaches of Achlioptas \& Iliopoulos and Harvey \& Vondr\'{a}k}
\label{achlioptas-sec}
Achlioptas \& Iliopoulos \cite{achlioptas} and Harvey \& Vondr\'{a}k \cite{harvey} gave generic frameworks for analyzing variants of the Moser-Tardos algorithm, applicable to different types of combinatorial configurations. These frameworks can include vertex-colorings, permutations, Hamiltonian cycles of graphs, spanning trees, matchings, and other settings. For the case of permutations, both of these frameworks give a version of the Swapping Algorithm and show that it terminates under the same conditions as we do, which in turn are the same conditions as the LLL (Theorem~\ref{thm:lopsided}).

The key difference between our approach and \cite{achlioptas,harvey} is that they enumerate the entire history of all resamplings to the permutations. In contrast, our proof is based on the Witness Tree Lemma; this is a much more succinct structure that ignores most of the resamplings, and only enumerates the few resamplings that are necessary to justify a single item in the execution log. Their proofs are much simpler than ours; a major part of the complexity of our proof lies in the need to argue that the bad-events which were ignored by the witness tree do not affect the probabilities. (The ignored bad-events \emph{do} interact with the variables we need to track for the witness tree, but do so in a ``neutral'' way.) 

\emph{If our only goal is to prove that the Swapping Algorithm terminates in polynomial time, then the other two frameworks give a better and simpler approach.}
However, the Witness Tree Lemma allows much more precise estimates for many types of events. The main reason for this precision is the following: suppose we want to show that some event $E$ has a low probability of occurring during or after the execution of the Swapping Algorithm. The proof strategy of Moser \& Tardos is to take a union-bound over all witness trees that correspond to this event. In this case, we are able to show a probability bound which is proportional to the total weight of all such witness trees. This can be a relatively small number as only the witness trees connected to $E$ are relevant. Our analysis, which is also based on witness trees, is able to show similar types of bounds.

However, the analysis of Achlioptas \& Iliopoulos and Harvey \& Vondr\'{a}k is not based on witness trees, but the much larger set of \emph{full execution logs}. The number of possible execution logs can be exponentially larger than the number of witness trees. It is very inefficient to take a union bound over all such logs. Hence, Achlioptas \& Iliopoulos and Harvey \& Vondr\'{a}k give bounds which are exponentially weaker (in a certain technical sense) than the ones we provide. 

Many properties of the Swapping Algorithm depend on the fine degree of control provided by the Witness Tree Lemma, and it seems difficult to obtain them from the alternate LLLL approaches. We list a few of these properties here.

\textbf{The LLL criterion without slack.} As a simple example of the problems caused by taking a union bound over execution logs, suppose that we satisfy the LLL criterion without slack, say $e p d = 1$; here, as usual, $p$ and $d$ are bounds respectively on the probability of any bad event and the degree of any bad event in the dependency graph. In this case, we show that the expected time for our Swapping Algorithm to terminate is $O(m)$. In contrast, in Achlioptas \& Iliopoulous or Harvey \& Vondr\'{a}k, they require satisfying the LLL criterion with slack $e p (1+\epsilon) d = 1$, and achieve a termination time of $O(m/\epsilon)$. They require this slack term in order to damp the exponential growth in the number of execution logs. (Harvey \& Vondr\'{a}k show that if the symmetric LLL criterion is satisfied without slack, then the Shearer criterion \cite{shearer:lll} is satisfied with slack $\epsilon = O(1/m)$. Thus, they would achieve a running time of $O(m^2)$ without slack.)

\textbf{Arbitrary choice of which bad-event to resample.} The Swapping Algorithm as we have stated it is actually under-determined, in that the choice of which bad-event to resample is arbitrary. In contrast, in both Achlioptas \& Iliopoulos and Harvey \& Vondr\'{a}k, there is a fixed priority on the bad-events. (The work of \cite{kolmogorov} has shown that this restriction can be removed in certain special cases of the Achlioptas \& Iliopoulous setting, including for random permutations and matchings.)  This freedom can be quite useful. For example, in Section~\ref{sec:parallel} we consider a parallel implementation of our Swapping Algorithm. We will select which bad-events to resample in a quite complicated and randomized way. However, the correctness of the parallel algorithm will follow from the fact that it simulates some serial implementation of the Swapping Algorithm.

\textbf{The Moser-Tardos distribution.} The Witness Tree Lemma allows us to analyze the so-called ``Moser-Tardos (MT) distribution,'' first discussed by \cite{haeupler}. The LLL and its algorithms ensure that bad-events $\mathcal B$ cannot possibly occur. In other words, we know that the configuration produced by the LLL has the property that no $B \in \mathcal B$ is true. In many applications of the LLL, we may wish to know more about such configurations, other than they exist.

 There are a variety of reasons we might want this; we give two examples for the ordinary, variable-based LLL. Suppose that we have some weights for the values of our variables, and we define the objective function on a solution $\sum_i w(X_i)$; in this case, if we are able to estimate the probability that a variable $X_i$ takes on value $j$ in the \emph{output} of the LLL (or Moser-Tardos algorithm), then we may be able to show that configurations with a good objective function exist. A second example is when the number of bad-events becomes too large, perhaps exponentially large. In this case, the Moser-Tardos algorithm cannot test them all. However, we may still be able to ignore a subset of the bad events, and argue that the probability that they are true at the end of the Moser-Tardos algorithm is small even though they were never checked. 

The Witness Tree Lemma gives us an extremely powerful result concerning this MT distribution, which carries over to the Swapping Algorithm.
\begin{proposition}
Let $E \equiv \pi_{k_1}(x_1) = y_1 \wedge \dots \wedge \pi_{k_r}(x_r) = y_r$. Then the probability that $E$ is true in the output of the Swapping Algorithm, is at most $P_{\Omega}(E) \prod_{B' \sim E} (1 + \mu(B'))$.
\end{proposition}
\begin{proof}
See \cite{haeupler} for the proof of this for the ordinary MT algorithm; the extension to the Swapping Algorithm is straightforward.
\end{proof}

\textbf{Bounds on the depth of the resampling process.} One key requirement for parallel variants of the Moser-Tardos algorithm appears to be that the resampling process has logarithmic depth. This is equivalent to showing that there are no deep witness trees. This follows easily from the Witness Tree Lemma, along the same lines as in the original paper of Moser \& Tardos, but appears to be very difficult in the other LLLL frameworks.

\textbf{Partial resampling.} In \cite{harris-srin2}, a partial resampling variant of the Moser-Tardos algorithm was developed. In this variant, one only resamples a small, random subset of the variables (or, in our case, permutation elements) which determine a bad-event. To analyze this variant, \cite{harris-srin2} developed an alternate type of witness tree, which only records the variables which were actually resampled. Ignoring the other variables can drastically prunes the space of witness trees. Again, this does not seem to be possible in other LLLL frameworks in which the \emph{full} execution log must be recorded. We will see an example of this in Theorem~\ref{szabo-thm}; we do not know of any way to show results such as Theorem~\ref{szabo-thm} using the frameworks of either Achlioptas \& Iliopoulos or Harvey \& Vondr\'{a}k.

\section{A parallel version of the Swapping Algorithm}
\label{sec:parallel}
The Moser-Tardos resampling algorithm for the ordinary LLL can be transformed into an RNC algorithm by allowing a slight slack in the LLL's sufficient condition \cite{moser-tardos}. The basic idea is that in every round, we select a \emph{maximal independent set} of bad-events to resample. Using the known distributed/parallel algorithms for MIS, this can be done in RNC; the number of resampling rounds is then shown to be logarithmic whp (``with high probability"), in 
\cite{moser-tardos}.

In this section, we will describe a parallel algorithm for the Swapping Algorithm, which runs along the same lines. However, everything is more complicated than in the case of the ordinary LLL. In the Moser-Tardos algorithm, events which are not connected to each other cannot affect each other in any way. For the permutation LLL, such events can interfere with each other, but do so rarely. Consider the following example. 
Suppose that at some point we have two active bad-events, ``$\pi_k(1) = 1$" and
``$\pi_k(2) = 2$" respectively, and so we decide to resample them simultaneously (since they are not
connected to each other, and hence constitute an independent set).
When we are resampling the bad-event $\pi_k(1) = 1$, we may swap $1$ with $2$; in this case, we are automatically fixing the second bad-event as well. The sequential algorithm, in this case, would only swap a single element. The parallel algorithm should likewise \emph{not} perform a second swap for the second bad-event, or else it would be over-sampling. Avoiding this type of conflict is quite tricky.

Let $n = n_1 + \dots + n_K$; since the output of the algorithm will be the contents of the permutations $\pi_1, \dots, \pi_k$, this algorithm should be measured in terms of $n$, and we must show that this algorithm runs in $\log^{O(1)} n$ time.  We will make the following assumptions in this section.  First, we assume that $|\mathcal B|$, the total number of potential bad-events, is polynomial in $n$. This assumption can be relaxed if we have the proper kind of ``separation oracle'' for $\mathcal B$. Next, we assume that every element $B \in \mathcal B$ has size $|B| \leq M = \log^{O(1)} n$; this holds in many cases.

We describe the following Parallel Swapping Algorithm:
\begin{enumerate}
\item In parallel, generate the permutations $\pi_1, \dots, \pi_N$ uniformly at random.
\item We proceed through a series of \emph{rounds} while there is some true bad-event. In round $i$ ($i = 1, 2, \ldots,$) do the following:
\begin{enumerate}
\item[(3)] Let $\mathcal V_{i,1} \subseteq \mathcal B$ denote the set of bad-events which are currently true at the beginning of round $i$. We will attempt to fix the bad-events in $\mathcal V_{i,1}$ through a series of \emph{sub-rounds}. This may introduce new bad-events, but we will not fix any newly created bad-events until round $i+1$. 

We repeat the following for $j = 1, 2, \dots$ as long as $\mathcal V_{i,j} \neq \emptyset$:
\begin{enumerate}
\item[(4)] Let $I_{i,j}$ be a maximal independent set (MIS) of bad-events in $\mathcal V_{i,j}$.
\item[(5)] For each true bad-event $B \in I_{i,j}$, choose the swaps corresponding to $B$. Namely, if we have some bad-event $B$ involving triples $(k_1, x_1, y_1), \dots, (k_r, x_r, y_r)$, then we select each $z_l \in [n_{k_l}]$, which is the element to be swapped with $\pi_{k_l}(x_l)$ according to procedure Swap. \emph{Do not perform the indicated swaps at this time though!} 
 We refer to $(k_1, x_1), \dots, (k_r, x_r)$ as the swap-sources of $B$ and the $(k_1, z_1)$, $\dots$, $(k_r, z_r)$ as the swap-mates of $B$.
\item[(6)]  Select a random ordering $\rho_{i,j}$ of the elements of $I_{i,j}$. Consider the graph $G_{i,j}$ whose vertices correspond to elements of $I_{i,j}$: add an edge connecting $B$ with $B'$ if $\rho_{i,j}(B) < \rho_{i,j}(B')$ \emph{and} one of the swap-mates of $B$ is a swap-source of $B'$.
Generate $I'_{i,j} \subseteq I_{i,j}$ as the \emph{lexicographically-first MIS} (LFMIS) of the resulting graph $G_{i,j}$, with respect to the vertex-ordering $\rho_{i,j}$.
\item[(7)] For each permutation $\pi_k$, enumerate all the transpositions $(x \ z)$ corresponding to elements of $I'_{i,j}$, arranged in order of $\rho_{i,j}$.   Say these transpositions are, in order $(x_1, z_1), \dots (x_l, z_l)$, where $l \leq n$. Compute, in parallel for all $\pi_k$, the composition $\pi_k' = \pi_k (x_l \ z_l) \dots (x_1 \ z_1)$.
\item[(8)] Update $\mathcal V_{i,j+1}$ from $\mathcal V_{i,j}$ by removing all elements which are either no longer true for the current permutation, \emph{or} are connected via $\sim$ to some element of $I'_{i,j}$.
\end{enumerate}
\end{enumerate}
\end{enumerate}

Most of the steps of this algorithm can be implemented using standard parallel algorithms. For example, step (1) can be performed simply by having each element of $[n_k]$ choose a random real and then executing a parallel sort. The independent set $I_{i,j}$ can be found in time in polylogarithmic time using \cite{abi:mis,luby:mis}. 

The difficult step to parallelize is in selecting the LFMIS $I'_{i,j}$. In general, the problem of finding the LFMIS is P-complete \cite{lfmis-p}, hence we do not expect a generic parallel algorithm for this. However, what saves us it that the ordering $\rho_{i,j}$ and the graph $G_{i,j}$ are constructed in a highly random fashion. 

This allows us to use the following greedy algorithm to construct $I'_{i,j}$, the LFMIS of $G_{i,j}$:
\begin{enumerate}
\item Let $H_1$  be the directed graph obtained by orienting all edges of $G_{i,j}$ in the direction of $\rho_{i,j}$. Repeat the following for $s = 1, 2, \dots,$:
\begin{enumerate}
\item[(2)] If $H_s = \emptyset$ terminate.
\item[(3)] Find all source nodes of $H_s$. Add these to $I'_{i,j}$.
\item[(4)]  Construct $H'_{s+1}$ by removing all source nodes and all successors of source nodes from $H'_s$.
\end{enumerate}
\end{enumerate}
The output of this algorithm is the LFMIS $I'_{i,j}$. Each step can be implemented in parallel time $O(1)$. The number of iterations of this algorithm is the length of the longest directed path in $G'_{i,j}$. So it suffices it show that, whp, all directed paths in $G'_{i,j}$ have length at most polylogarithmic in $n$.

\begin{proposition}
Let $I \subseteq \mathcal B$ be an an arbitrary independent set of true bad-events, and suppose all elements of $\mathcal B$ have size $\leq M$. Let $G = G_{i,j}$ be the graph constructed in Step (6) of the Parallel Swapping Algorithm.

Then whp, every directed path in $G$ has length $O(M + \log n)$.
\end{proposition}
\begin{proof}
 One of the main ideas below is to show that for the \emph{typical} 
$B_1, \dots, B_l \in I$, where $l = 5(M + \log n)$, the probability that $B_1, \dots, B_l$ form a directed path is small.  Suppose we select $B_1, \dots, B_l \in I$ uniformly at random without replacement. Let us analyze how these could form a directed path in $G$. (We may assume $|I| > l$ or otherwise the result holds trivially.)

First, it must be the case that $\rho(B_1) < \rho(B_2) < \dots < \rho(B_l)$. This occurs with probability $1/l!$.

Next, it must be that the swap-mates of $B_s$ overlap the swap-sources of $B_{s+1}$, for $s = 1, \dots, l-1$. Now, $B_s$ has $O(M)$ swap-mates; each such swap-mate can overlap with at most one element of $I$, since $I$ is an independent set.  Conditional on having chosen $B_1, \dots, B_s$, there are a remaining $|I|-s$ choices for $B_{s+1}$. This gives that the probability of having $B_s$ with an edge to $B_{s+1}$, conditional on the previous events, is at most $\frac{M}{|I|-s}$. (The fact that swap-mates are chosen randomly does not give too much of an advantage here.)

Putting this all together, the total probability that there is a directed path on $B_1, \dots, B_l$ is
$$
P(\text{directed path $B_1, \dots, B_l$}) \leq \frac {M^{l-1} ( |I| - l ) !}{ ( |I|-1 )! l!}
$$

Since the above was for a random $B_1, \dots, B_l$, the probability that there is \emph{some} such path (of length $l$) is at most
{\allowdisplaybreaks
\begin{align*}
P(\text{some directed path}) & \leq  \frac{|I|!}{(|I| - l)!} \times \frac{M^{l-1} (|I| - l)!}{(|I|-1)! l!} \\
&= |I| \times  \frac{M^{l-1}}{l!} \leq  n \times  \frac{M^{l-1}}{(l/e)^l} \leq n^{-\Omega(1)},
\end{align*}
}
since  $l = 5(M + \log n)$. 
\end{proof}

So far, we have shown that each sub-round of the Parallel Swapping Algorithm can be executed in parallel time $\log^{O(1)} n$. Next, we show that whp that number of sub-rounds corresponding to any round is bounded by $\log^{O(1)} n$.

\begin{proposition}
Suppose $|\mathcal B| = n^{O(1)}$ and all elements $B \in \mathcal B$ have size $|B| \leq M$. Then whp, we have $\mathcal V_{i,j} = \emptyset$ for some $j = O(M \log^2 n)$.
\end{proposition}
\begin{proof}
We will first show the following: Let $B \in I$, where $I$ is an arbitrary independent set of $\mathcal B$. Then with probability at least $1 - \frac{1}{2 M \ln n}$ we have $B \in I'$ as well, where $I'$ is the LFMIS associated with $I$.

Observe that if there is no $B' \in I$ such that $\rho(B') < \rho(B)$ and such that a swap-mate of $B'$ overlaps with a swap-source of $B$, then $B \in I'$ (this is not a necessary condition). We will analyze the ordering $\rho$ using the standard trick, in which each element $B \in I$ chooses a rank $W(B) \sim \text{Uniform}[0,1]$, independently and identically. The ordering $\rho$ is then formed by sorting in increasing ordering of $W$. In this way, we are able to avoid the dependencies induced by the rankings. For the moment, let us suppose that the rank $W(B)$ is \emph{fixed} at some real value $w$. We will then count how many $B' \in I$ satisfy $W(B') < w$ and a swap-mate of $B'$ overlaps a swap-source of $B$.

So, let us consider some swap-source $s$ of $B$ in permutation $k$, and consider some $B'_j \in I$ which has $r'_j$ other elements in permutation $k$. For $l = 1, \dots, r'_j$, there are $n_k - l + 1$ possible choices for the $l^{\text{th}}$ swap-mate from $B'_j$, and hence the total expected number of swap-mates of $B'$ which overlap $s$ is at most
\begin{align*}
\bE[ \text{ \# swap-mates of $B'_j$ overlapping $s$} ] &\leq \sum_{l=1}^{r'_j} \frac{1}{n_k - l + 1} \\
&\leq \int_{l=1}^{r'_j+1} \frac{1}{n_k - l + 1} dl \\
&=\ln (\frac{n_k}{n_k-r'_j})
\end{align*}

Next, sum over all $B'_j \in I$. Observe that since $I$ is an independent set, we must have $\sum r'_j \leq n_k-1$. Thus, 
\begin{align*}
\bE[ \text{ \# swap-mates of some $B'_j$ overlapping $s$} ] &\leq \sum_j \ln(\frac{n_k}{n_k-r'_j}) \\
 &\leq \ln(\frac{n_k}{n_k-\sum_j r'_j}) \qquad \text{by concavity} \\
 &\leq \ln n_k \leq \ln n \\
\end{align*}

Thus, summing over all swap-sources of $B$, the total probability that there is some $B'$ with $\rho(B') \leq B$ and for which a swap-mate overlaps a swap-source of $B$, is at most $w |B| \ln n \leq w M \ln n$. By Markov's inequality, we have
$$
P( B' \in I' \mid W(B) = w ) \geq 1 - w M \ln n
$$

Integrating over $w$, we have that $B' \in I'$ with probability at least
$$
P( B' \in I'  ) \geq 1 - \frac{1}{2 M \ln n}
$$

Now, using this fact, we show that $\mathcal V_{i,j}$ is decreasing quickly in size.  For, suppose $B \in \mathcal V_{i,j}$. So $B \sim B'$ for some $B' \in I_{i,j}$, as $I_{i,j}$ is a maximal independent set (possibly $B = B'$). We will remove $B$ from $\mathcal V_{i,j+1}$ if $B' \in I'_{i,j}$, which occurs with probability at least  $1 - \frac{1}{2 M \ln n}$. As $B$ was an arbitrary element of $\mathcal V_{i,j}$, this shows that $\bE\bigl[|\mathcal V_{i,j+1}| \mid \mathcal V_{i,j}\bigr] \leq (1 - \frac{1}{2 M \ln n}) |\mathcal V_{i,j}|$. 

For $j = \Omega(M \log^2 n)$,  this implies that 
\begin{align*}
\bE\bigl[ |\mathcal V_{i,j}| \bigr] &\leq  (1 - \frac{1}{2 M \ln n})^{\Omega(M \log^2 n)} |\mathcal V_{i,1}| \leq n^{-\Omega(1)}
\end{align*}

This in turn implies that $\mathcal V_{i,j} = \emptyset$ with high probability, for $j = \Omega( M \log^2 n )$.
\end{proof}

To finish the proof, we must show that the number of rounds is itself bounded whp.  We begin by showing that Witness Tree Lemma remains valid in the parallel setting.
\begin{proposition}
When we execute this parallel swapping algorithm, we may generate an ``execution log'' according to the following rule: suppose that we resample $B$ in round $i,j$ and $B'$ in round $i', j'$. Then we place $B$ before $B'$ iff:
\begin{enumerate}
\item $i < i'$; OR
\item $i = i'$ AND $j < j'$; OR
\item $i = i'$ and $j = j'$ and $\rho_{i,j}(B) < \rho_{i', j'} (B')$
\end{enumerate}
that is, we order the resampled bad-events lexicographically by round, sub-round, and then rank $\rho$.

Given such an execution log, we may also generate witness trees in the same manner as the sequential algorithm.

Now let $\tau$ be any witness tree; we have
$$
P(\text{$\tau$ appears}) \leq \prod_{B \in \tau} P_{\Omega}(B)
$$
\end{proposition}
\begin{proof}
Observe that the choice of swaps for a bad-event $B$ at round $i$, subround $j$, and rank $\rho_{i,j}(B)$, is only affected by the events in earlier rounds / subrounds as well as other $B' \in I_{i,j}$ with $\rho_{i,j}(B') < \rho_{i,j}(B)$. 

Thus, we can view this parallel algorithm as simulating the sequential algorithm, with a particular rule for selecting the bad-event to resample. Namely, we keep track of the sets $\mathcal V_i$ and $I_{i,j}$ as we do for the parallel algorithm, and within each sub-round we resample the bad-event in $I_{i,j}$ with the minimum value of $\rho_{i,j}(B)$.

This is why it is critical in step (6) that we select $I'_{i,j}$ to be the lexicographically-first MIS; this means that the presence of $B \in I'_{i,j}$ cannot be affected with $B'$ with $\rho(B') > \rho(B)$. 
\end{proof}
\begin{proposition}
Let $B$ be any resampling performed at the $i^{\text{th}}$ round of the Parallel Swapping Algorithm (that is, $B \in I'_{i,j}$ for some integer $j > 0$) Then the witness tree corresponding to the resampling of $B$ has height exactly $i$.
\end{proposition}
\begin{proof}
First, note that if we have $B \sim B'$ in the execution log, where $B$ occurs earlier in time, and the witness tree corresponding to $B$ has height $i$, then the witness tree corresponding to $B'$ must have height $i+1$. So it will suffice to show that if $B \in I'_{i,j}$, then we must have $B \sim B'$ for some $B' \in I'_{i-1, j'}$.

At the beginning of round $i$, it must be the case that $\pi^{i}$ makes the bad-event $B$ true. By Proposition~\ref{change-prop}, either the bad-event $B$ was already true at the beginning of round $i-1$, or some bad-event $B' \sim B$ was resampled at round $i-1$. If it is the latter, we are done. 

So suppose $B$ was true at the beginning of round $i-1$. So $B$ was an element of $\mathcal V_{i-1,1}$.  In order for $B$ to have been removed from $\mathcal V_{i-1}$, then either we had $B \sim B' \in I'_{i-1, j'}$, in which case we are also done, or after some sub-round $j'$ the event $B$ was no longer true. But again by Proposition~\ref{change-prop}, in order for $B$ to become true again at the beginning of round $i$, there must have been some bad-event $B' \sim B$ encountered later in round $i-1$.
\end{proof}

This gives us the key bound on the running time of the Parallel Swapping Algorithm. We give only a sketch of the proof, since the argument is identical to that of \cite{moser-tardos}.

\begin{proposition}
Suppose  that $\epsilon > 0$ and that there is some assignment of weights $\mu: \mathcal B \rightarrow [0, \infty)$ which satisfies, for every $B \in \mathcal B$, the condition
$$
\mu(B) \geq (1 +\epsilon) P_{\Omega}(B) \prod_{B' \sim B} (1 + \mu(B'))
$$
Then, whp, the Parallel Swapping Algorithm terminates after $\frac{\log^{O(1)} (n \sum_B \mu(B))}{\epsilon}$ rounds. 
\end{proposition}
\begin{proof}
Consider the event that for some $B \in \mathcal B$, that $B$ is resampled after $i$ rounds of the Parallel Swapping Algorithm. In this case, $\hat \tau$ has height $i$. As shown in \cite{moser-tardos}, the sum, over all witness trees of some height $h$, of the product of the probabilities of the constituent events in the witness trees, is decreasing exponentially in $h$. So, for any fixed $B$, the probability that this occurs is exponentially small; this remains true after taking a union-bound over the polynomial number of $B \in \mathcal B$.
\end{proof}

We can put this analysis all together to show:

\begin{theorem}
\label{parallel-thm}
Suppose $|\mathcal B| = n^{O(1)}$ and that for all $B \in \mathcal B'$ we have $|B| \leq \log^{O(1)} n$. Suppose also that $\epsilon > 0$ and that there is some assignment of weights $\mu: \mathcal B \rightarrow [0, \infty)$ which satisfies, for every $B \in \mathcal B$, the condition
$$
\mu(B) \geq (1 +\epsilon) P_{\Omega}(B) \prod_{B' \sim B} (1 + \mu(B'))
$$
Then, whp, the Parallel Swapping Algorithm terminates after $\frac{\log^{O(1)} (n \sum_B \mu(B)) }{\epsilon}$ time.
\end{theorem}
\begin{proof}
The number of rounds, the number of sub-rounds per round, and the running time of each sub-round, are all polylogarithmic in $n$ whp.
\end{proof}

\section{Algorithmic Applications}
\label{alg-sec}
The LLL for permutations plays a role in diverse combinatorial constructions. Using our algorithm, nearly all of these constructions become algorithmic.  We examine a few selected applications now.

\subsection{Latin transversals}
Suppose we have an $n \times n$ matrix $A$. The entries of this matrix come from a set $C$ which are referred to as \emph{colors}. A \emph{Latin transversal} of this matrix is a permutation $\pi \in S_n$, such that no color appears twice among the entries $A(i, \pi(i))$; that is, there are no $i \neq j$ with $A(i, \pi(i)) = A(i', \pi(i'))$. A typical question in this area is the following: suppose each color $c$ appears at most $\Delta$ times in the matrix. How large can $\Delta$ be so as to guarantee the existence of a Latin transversal?

In \cite{erdos-spencer}, a proof using the probabilistic form of the Lov\'{a}sz Local Lemma for permutations was given, showing that $\Delta \leq n/(4e)$ suffices. This was the first application of the LLL to permutations. This bound was subsequently improved by \cite{bissacot} to the criterion $\Delta \leq (27/256) n$; this uses a variant of the probabilistic Local Lemma which is essentially equivalent to Pegden's variant on the constructive Local Lemma.
Using our algorithmic LLL, we can almost immediately transform the existential proof of \cite{bissacot} into a constructive algorithm. To our knowledge, this is the first polynomial-time algorithm for constructing such a transversal.

\begin{theorem}
Suppose $\Delta \leq (27/256) n$. Then there is a Latin transversal of the matrix. Furthermore, the Swapping Algorithm selects such a transversal in polynomial time.
\end{theorem}
\begin{proof}
For any quadruples $i,j, i', j'$ with $A(i,j) = A(i', j')$, we have a bad-event $(i, j), (i', j')$. Such an event has probability $\tfrac{1}{n(n-1)}$. We give weight $\mu(B) = \alpha$ to every bad event $B$, where $\alpha$ is a scalar to be determined. 

This bad-event can have up to four types of neighbors $(i_1, j_1, i'_1, j'_1)$, which overlap on one of the four coordinates $i, j, i', j'$; as discussed in \cite{bissacot}, all the neighbors of any type are themselves neighbors in the dependency graph. Since these are all the same, we will analyze just the first type of neighbor, one which shares the same value of $i$, that is $i_1= i$. We now may choose any value for $j_1$ ($n$ choices). At this point, the color $A(i_1, j_1)$ is determined, so there are $\Delta-1$ remaining choices for $i'_1, j'_1$. 

By Lemma~\ref{witness-tree-lemma} and Pegden's criterion \cite{pegden}, a sufficient condition for the convergence of the Swapping Algorithm is that
$$
\alpha \geq \frac{1}{n(n-1)} (1 + n (\Delta - 1) \alpha)^4
$$

Routine algebra shows that this has a positive real root $\alpha$ when $\Delta \leq (27/256) n$.
\end{proof}

In \cite{szabo:transversals}, Szab\'{o} considered a generalization of this question: suppose that we seek a transversal, such that no color appears more than $s$ times. When $s = 1$, this is asking for a Latin transversal. Szab\'{o} gave similar criteria ``$\Delta \leq \gamma_s n$" for $s$ a small constant. Such bounds can be easily obtained constructively using the permutation LLL as well.

By combining the permutation LLL with the partial resampling approach of \cite{harris-srin2}, we can provide asymptotically optimal bounds for large $s$:
\begin{theorem}
\label{szabo-thm}
Suppose $\Delta \leq (s - c \sqrt{s}) n$, where $c$ is a sufficiently large constant. Then there is a transversal of the matrix in which each color appears no more than $s$ times. This transversal can be constructed in polynomial time.
\end{theorem}
\begin{proof}
For each set of $s$ appearances of any color, we have a bad event. We use the partial resampling framework, to associate the fractional hitting set which assigns weight $\binom{s}{r}^{-1}$ to any $r$ appearances of a color, where $r = \lceil \sqrt{s} \rceil$.

We first compute the probability of selecting a given $r$-set $X$. From the fractional hitting set, this has probability $\binom{s}{r}^{-1}$. In addition, the probability of selecting the indicated cells is $\frac{(n-r)!}{n!}$. So we have $p \leq \binom{s}{r}^{-1} \frac{(n-r)!}{n!}$.

Next, we compute the dependency of the set $X$. First, we may select another $X'$ which overlaps with $X$ in a row or column; the number of such sets is $2 r n \binom{\Delta}{r-1}$. Next, we may select any other $r$-set with the same color as $X$ (this is the dependency due to $\bowtie$ in the partial resampling framework; see \cite{harris-srin2} for more details). The number of such sets is $\binom{\Delta}{r}$.

So the LLL criterion is satisfied if
$$
e \times  \binom{s}{r}^{-1} \frac{(n-r)!}{n!} \times \Bigl( 2 r n \binom{\Delta}{r-1} + \binom{\Delta}{r} \Bigr) \leq 1
$$

Simple calculus now shows that this can be satisfied when $\Delta \leq (s - O(\sqrt{s})) n$. Also, it is easy to detect a true bad-event and resample it in polynomial time, so this gives a polynomial-time algorithm.
\end{proof}

Our result depends on the Swapping Algorithm in a fundamental way --- it does not follow from Theorem~\ref{thm:lopsided} (which would roughly require  $\Delta \leq (s/e) n$). Hence, prior to this paper, we would not have been able to even show the existence of such transversals; here we provide an efficient algorithm as well.
To see that our bound is asymptotically optimal, consider a matrix in which the first $s+1$ rows all contain a given color, a total multiplicity of $\Delta = (s+1) n$. Then the transversal must contain that color at least $s+1$ times.

\subsection{Rainbow Hamiltonian cycles and related structures}
The problem of finding Hamiltonian cycles in the complete graph $K_n$, with edges of distinct colors, was first studied in \cite{hahn-thomassen}. This problem is typically phrased in the language of graphs and edges, but we can rephrase it in the language of Latin transversals, with the additional property that the permutation $\pi$ has full cycle. How often can a color appear in the matrix $A$, for this to be possible? In \cite{albert}, it was shown that such a transversal exists if each color appears at most $\Delta = n/32$ times.\footnote{The terminology used for rainbow Hamilton cycles is slightly different from that of Latin transversals. In the context of Hamilton cycles, one often assumes that the matrix $A$ is symmetric. Furthermore, since $A(x,y)$ and $A(y,x)$ always have the same color, one only counts this as a single occurrence of that color. Thus, for example, in \cite{albert}, the stated criterion is that the matrix $A$ is symmetric and a color appears at most $\Delta/64$ times.} This proof is based on applying the non-constructive Lov\'{a}sz Local Lemma to the probability space induced by a random choice of full-cycle permutation.
This result was later generalized in \cite{kriv}, to show the following result: if each color appears at most $\Delta \leq c_0 n$ times for a certain constant $c_0 > 0$, then not only is there a full-cycle Latin transversal, but there are also cycles of each length $3 \leq k \leq n$. The constant $c_0$ was somewhat small, and this result was also non-constructive.
Theorem~\ref{thm:rainbow-cycles} uses the Swapping Algorithm to construct Latin transversals with essentially arbitrary cycle structures; this generalizes \cite{kriv} and \cite{albert} quite a bit. 

\begin{theorem}
\label{thm:rainbow-cycles}
Suppose that each color appears at most $\Delta \leq 0.027 n$ times in the matrix $A$, and $n$ is sufficiently large.
Let $\tau$ be any permutation on $n$ letters, whose cycle structure contains no fixed points nor swaps (2-cycles). Then there is a Latin transversal $\pi$ which is conjugate to $\tau$ (i.e., has the same cycle structure); furthermore the Swapping Algorithm finds it in polynomial time. Also, the Parallel Swapping Algorithm finds it in time $\log^{O(1)} n$.
\end{theorem}
\begin{proof}
We cannot apply the Swapping Algorithm directly to the permutation $\pi$, because we will not be able to control its cycle structure. Rather, we will set $\pi = \sigma^{-1} \tau \sigma$, and apply the Swapping Algorithm to $\sigma$.

A bad-event is that $A(x, \pi(x)) = A(x', \pi(x'))$ for some $x \neq x'$. Using the fact that $\tau$ has no fixed points or 2-cycles, we can see that this is equivalent to one of the following two situations:  (A) There are $i, i', x, y, x', y'$ such that $\sigma(x) = i, \sigma(y) = \tau (i), \sigma(x') = i', \sigma(y') = \tau(i')$, and $x, y, x', y'$ are distinct, and $i, i', \tau(i), \tau(i')$ are distinct, and $A(x,y) = A(x', y')$ or (B) There are $i, x, y, z$ with $\sigma(x) = i, \sigma(y) = \tau(i), \sigma(z) = \tau^2(i)$, and all of $x,y,z$ are distinct, and $A(x,y) = A(y,z)$. We will refer to the first type of bad-event as an event of type A led by $i$ (such an event is also led by $i'$); we will refer to the second type of bad-event as type B led by $i$.

Note that in an A-event, the color is repeated in distinct column and rows, and in a B-event the column of one coordinate is the row of another. So, to an extent, these events are mutually exclusive. Much of the complexity of the proof lies in balancing the two configurations. To a first approximation, the worst case occurs when A-events are maximized and B-events are impossible. This intuition should be kept in mind during the following proof.

We will define the function $\mu$ as follows. Each event of type A is assigned the same weight $\mu_A$, and each event of type B is assigned weight $\mu_B$. The event of type A has probability $(n-4)!/n!$ and each event of type B has probability $(n-3)!/n!$. In the following proof, we shall need to compare the relative magnitude of $\mu_A, \mu_B$. In order to make this concrete, we set
$$
\mu_A = 2.83036 n^{-4}, \mu_B = 1.96163 n^{-3}
$$
(In deriving this proof, we left these constant coefficients undetermined until the end of the computation, and we then verified that all desired inequalities held.)

Now, to apply Pegden's criterion \cite{pegden} for the convergence of the Swapping Algorithm, we will need to analyze the independent sets of neighbors each bad-event can have in the dependency graph. In order to keep track of this neighborhood structure, it will be convenient to define the following sums. We let $t$ denote the sum of $\mu(X)$ over all bad-events $X$ involving some fixed term $\sigma(x)$. Let $s$ denote the sum of $\mu(X)$ over all bad-events $X$ (of type either A or B) led by some fixed value $i$, and let $b$ denote the sum of $\mu(X)$ over B-events $X$ alone. Recall that each bad-event of type A is led by $i$ and also by $i'$. 

We now examine how to compute the term $t$. Consider a fixed value $x$; we will enumerate all the bad-events that involve $\sigma(x)$. These correspond to color-repetitions involving either row or column $x$ in the matrix $A$. Let $c_i$ (respectively $r_i$) denote the number of occurrences of color $i$ in column (respectively row) $x$ of the matrix, excluding $A(x,y)$ itself. 

We can have a color repetition of the form $A(y,x) = A(x,y')$ where $y \neq y'$; or we can have repetitions of the form $A(x,y) = A(x', y')$ or $A(y, x) = A(y', x')$, where $x \neq x', y \neq y'$ (but possibly $x' = y$). The total number of repetitions of the first type is $v_1 \leq \sum_i c_i r_i$. The total number of repetitions of the  second type is at most $v_2 \leq \sum_i c_i (\Delta - c_i - r_i)$. The total number of repetitions of the third type is at most $v_3 \leq \sum_i r_i (\Delta - c_i - r_i)$.

For a repetition of the first type, this must correspond to an B-event, in which $\sigma(y) = i, \sigma(x) = \tau(i), \sigma(y') = \tau^2(i)$ for some $i$. For a repetition of the second type, if $x' \neq y$ this correspond to an A-event in which $\sigma(x) = i, \sigma(y) = \tau(i), \sigma(x') = i', \sigma(y') = \tau(i')$ for some $i, i'$ \emph{or} alternatively if $x' = y$ it correspond to a B-event in which $\sigma(x) = i, \sigma(y) = \tau(i), \sigma(y') = \tau^2(i)$ for some $i$. A similar argument holds for the third type of repetition.

Summing all these cases, we have 
{\allowdisplaybreaks
\begin{align*}
t &\leq v_1 n \mu_B + v_2 (\max(n^2 \mu_A + n \mu_B)) + v_3 (\max(n^3 \mu_A + n \mu_B)) \\
&\leq v_1 n \mu_B + v_2 n^2 \mu_A + v_3 n^2 \mu_A  \\
&\leq \sum_j (c_j r_j n \mu_B + c_j (\Delta - c_j - r_j) n^2 \mu_A + r_j (\Delta - c_j - r_j) n^2 \mu_A)
\end{align*}
}

Observe that the the RHS is maximized when there are $n$ distinct colors with $c_j = 1$ and $n$ distinct colors with $r_j = 1$. For, suppose that a color has (say) $c_j > 1$. If we decrement $c_j$ by 1 while adding a new color with $c_{j'} = 1$, this changes the RHS by $(-1 + 2 (c_j + r_j) - \Delta) n^2 \mu_A + (-1 + \Delta - r_j) n \mu_B \geq 0$.

This gives us 
$$
t \leq 2 n^3 \Delta \mu_A
$$

Similarly, let us consider $s$. Given $i$, we choose some $y$ with $\sigma(y) = \tau(i)$. Now, we again list all color repetitions $A(x,y) = A(x',y')$ or $A(x,y) = A(y,z)$. The number of the former is at most $\sum_j c_j (\Delta - c_j - r_j)$ and the number of the latter is at most $\sum_j c_j r_j$. As before, this is maximized when each color appears once in the column, leading to 
$$
s \leq n^3 \Delta \mu_A
$$

For term $b$, the worst case is when each color appears $\Delta/2$ times in the row and column of $y$; this yields
$$
b \leq n^2 (\Delta/2) \mu_B
$$

Now consider a fixed bad-event A, with parameters $i, i', x,y,x', y'$, and let us count the sum over all independent sets of neighbors, of $\mu$. This could have one or zero children involving $\sigma(x)$ and similarly for $\sigma(y),\sigma(x'), \sigma(y')$; this gives a total contribution of $(1+t)^4$. The children could also overlap on $i$; the total set of possibilities is either zero children, a B-child led by $i-2$, a B-child led by $i-2$ and a child led by $i$, a child led by $i-1$, a child led by $i-1$ and a child led by $i+1$, a child led by $i$, a child led by $i+1$. There is an identical factor for the contributions of bad-events led by $i'-2, \dots, i'+1$.  In total, the criterion for A is that we must have
$$
\mu_A \geq \frac{(n-4)!}{n!} (1+t)^4 (1 + b + s b + s + s^2 + s + s)^2
$$

Applying the same type of analysis to an event of type B gives us the criterion:
$$
\mu_B \geq \frac{(n-3)!}{n!} (1+t)^3 (1 + b + s b + s b+ s + s^2 + s^2 + s + s^2 + s)
$$

Putting all these constraints together gives a complicated system of polynomial equations, which can be solved using a symbolic algebra package. Indeed, the stated values of $\mu_A, \mu_B$ satisfy these conditions when $\Delta \leq 0.027n$ and $n$ is sufficiently large.

Hence the Swapping Algorithm terminates, resulting in the desired permutation $\pi = \sigma^{-1} \tau \sigma$. It is easy to see that the Parallel Swapping Algorithm works as well.
\end{proof}
We note that for certain cycle structures, namely the full cycle $\sigma = (1 2 3 \dots n-1 \ n)$ and $n/2$ transpositions $\sigma = (1 2) (3 4) \dots (n-1 \ n)$, one can apply the LLLL directly to the permutation $\pi$. This gives a qualitatively similar condition, of the form $\Delta \leq c n$, but the constant term is slightly better than ours. For some of these settings, one can also apply a variant of the Moser-Tardos algorithm to find such permutations \cite{achlioptas}. However, these results do not apply to general cycle structures, and they do not give parallel algorithms.

\subsection{Strong chromatic number of graphs}
Suppose we have a graph $G$, with a given partition of the vertices into $k$ blocks each of size $b$, i.e., $V = V_1 \sqcup \dots \sqcup V_k$. We would like to $b$-color the vertices, such that every block has exactly $b$ colors, and such that no edge has both endpoints with the same color (i.e., it is a proper vertex-coloring). This is referred to as a \emph{strong coloring} of the graph. If this is possible for \emph{any} such partition of the vertices into blocks of size $b$, then we say that the graph $G$ has strong chromatic number $b$.

A series  of papers \cite{alon1992,axenovich,fellows,haxell2004} have provided bounds on the strong chromatic number of graphs, typically in terms of their maximum degree $\Delta$. In \cite{haxell2008}, it is shown that when $b \geq (11/4) \Delta + \Omega(1)$, such a coloring exists; this is the best bound currently known. Furthermore, the constant $11/4$ cannot be improved to any number strictly less than $2$. The methods used in most of these papers are highly non-constructive, and do not provide algorithms for generating such colorings. 

In this section, we examine two routes to constructing strong colorings. The first proof, based on \cite{aharoni}, builds up the coloring vertex-by-vertex, using the ordinary LLL. The second proof uses the permutation LLL to build the strong coloring directly. 
The latter appears to be the first RNC algorithm with a reasonable bound on $b$.  

We first develop a related concept to the strong coloring known as an \emph{independent transversal}. In an independent transversal, we choose a single vertex from each block, so that the selected vertices form an independent set of the graph.
\begin{proposition}
\label{it-prop}
Suppose $b \geq 4 \Delta$. Then $G$ has an independent transversal, which can be found in expected time $O(n \Delta)$.

Furthermore, let $v \in G$ be any fixed vertex. Then $G$ has an independent transversal which includes $v$, which can be found in expected time $O(n \Delta^2)$.
\end{proposition}
\begin{proof}
Use the ordinary LLL to select a single vertex uniformly from each block. See \cite{bissacot}, \cite{harris-srin2} for more details. This shows that, under the condition $b \geq 4 \Delta$, an independent transversal exists and is found in expected time $O(n \Delta)$. 

To find an independent transversal including $v$, we imagine assigning a weight $1$ to vertex $v$ and weight zero to all other vertices. As described in \cite{harris-srin2}, the expected weight of the independent transversal returned by the Moser-Tardos algorithm, is at least $\Omega(w(V)/\Delta)$, where $w(V)$ is the total weight of all vertices. This implies that that vertex $v$ is selected with probability $\Omega(1/\Delta)$. Hence, after running the Moser-Tardos algorithm for $O(\Delta)$ separate independent executions, one finds an independent transversal including $v$.
\end{proof}

Using this as a building block, we can form a strong coloring by gradually adding colors:
\begin{theorem}
\label{strong-color-thm2}
Suppose $b \geq 5 \Delta$. Then $G$ has a strong coloring, which can be found in expected time $O(n^2 \Delta^2)$.
\end{theorem}
\begin{proof}
(This proof is almost identical to the proof of Theorem 5.3 of \cite{aharoni}). We maintain a \emph{partial coloring} of the graph $G$, in which some vertices are colored with $\{1, \dots, b\}$ and some vertices are uncolored. Initially all vertices are uncolored. We require that in a block, no vertices have the same color, and no adjacent vertices have the same color. 

Now, suppose some color is partially missing from the strong coloring; say without loss of generality there is a vertex $w$ missing color $1$. In each block $i = 1,\dots, k$, we will select some vertex $v_i$ to have color 1. If the block does not have such a vertex already, we will simply assign $v_i$ to have color 1. If the block $i$ \emph{already} had some vertex $u_i$ with color $1$, we will swap the colors of $v_i$ and $u_i$ (if $v_i$ was previously uncolored, then $u_i$ will become uncolored).

We need to ensure three things. First, the vertices $v_1, \dots, v_k$ must form an independent transversal of $G$. Second, if we select vertex $v_i$ and swap its color with $u_i$, this cannot cause $u_i$ to have any conflicts with its neighbors. Third, we insist of selecting $w$ itself for the independent traversal. 

A vertex $u_i$ will have conflicts with its neighbors if $v_i$ currently has the same color as one of the neighbors of $u_i$. In each block, there are at least $b - \Delta$ possible choices of $v_i$ that avoid that; we must select an independent transversal among these vertices, which also includes the designated vertex $w$. By Proposition~\ref{it-prop}, this can be done in time $O(n^2 \Delta^2)$ as long as $b \geq 4 \Delta$.

Whenever we select the independent transversal $v_1, \dots, v_k$, the total number of colored vertices increases by at least one: for, the vertex $w$ becomes colored while it was not initially, and in every other block the number of colored vertices does not decrease. So, after $n$ iterations, the entire graph has a strong coloring; the total time is $O(n^2 \Delta^2)$.
\end{proof}

The algorithm based on the ordinary LLL is slow and is inherently sequential. Using the permutation LLL, one can obtain a more direct and faster construction; however, the hypothesis of the theorem will need to be slightly stronger.
\begin{theorem}
Suppose we have a given graph $G$ of maximum degree $\Delta$, whose vertices are partitioned into blocks of size $b$. Then if $b \geq \frac{256}{27} \Delta$, it is possible to strongly color graph $G$ in expected time $O(n \Delta)$. If $b \geq (\frac{256}{27} + \epsilon) \Delta$ for some constant $\epsilon > 0$, there is an RNC algorithm to construct such a strong coloring. 
\end{theorem}
\begin{proof}
We will use the permutation LLL. For each block, we assume the vertices and colors are identified with the set $[b]$. Then any proper coloring of a block corresponds to a permutation of $S_b$. When we discuss the color of a vertex $v$, we refer to $\pi_k(v)$ where $k$ is the block containing vertex $v$.

For each edge $f = \langle u, v \rangle \in G$ and any color $c \in [1, \dots b]$, we have a bad-event that both $u$ and $v$ have color $c$. (Note that we cannot specify simply that $u$ and $v$ have the \emph{same color}; because we have restricted ourselves to \emph{atomic} bad-events, we must list every possible color $c$ with a separate bad event.)

Each bad-event has probability $1/b^2$. We give weight $\mu(B) = \alpha$ to every bad event, where $\alpha$ is a scalar to be determined. 

Now, each such event $(u,v,c)$ is dependent with four other types of bad-events:
\begin{enumerate}
\item An event $u, v', c'$ where $v'$ is connected to vertex $u$;
\item An event $u', v, c'$ where $u'$ is connected to vertex $v$;
\item An event $u', v', c$ where $u'$ is in the block of $u$ and $v'$ is connected to $u'$;
\item An event $u', v', c$ where $v'$ is in the block of $v$ and $u'$ is connected to $v'$
\end{enumerate}

 There are $b \Delta$ neighbors of each type.  For any of these four types, all the neighbors are themselves connected to each other. Hence an \emph{independent} set of neighbors of the bad-event $(u,v,c)$ can contain one or zero of each of the four types of bad-events. 

Using Lemma~\ref{witness-tree-lemma} and Pegden's criterion \cite{pegden}, a sufficient condition for the convergence of the Swapping Algorithm is that
$$
\alpha \geq (1/b^2) \cdot (1 + b \Delta \alpha)^4
$$

When $b \geq \frac{256}{27} \Delta$, this has a real positive root $\alpha^*$ (which is a complicated algebraic expression). Furthermore,  in this case the expected number of swaps of each permutation is $\leq b^2 \Delta \alpha^* \leq \frac{256}{81} \Delta$. So the Swapping Algorithm terminates in expected time $O(n \Delta)$. A similar argument applies to the parallel Swapping Algorithm.
\end{proof}

\subsection{Hypergraph packing}
In \cite{random-inj}, the following packing problem was considered. Suppose we are given two $r$-uniform hypergraphs $H_1, H_2$ and an integer $n$. Is it possible to find two injections $\phi_i: V(H_i) \rightarrow [n]$ with the property that $\phi_1(H_1)$ is edge-disjoint to $\phi_2(H_2)$? (That is, there are no edges $e_1 \in H_1, e_2 \in H_2$ with $\{ \phi_1(v) \mid v \in e_1 \} = \{ \phi_2(v) \mid v \in e_2 \}$. ). A sufficient condition on $H_1, H_2, n$ was given using the LLLL. We achieve this algorithmically as well:

\begin{theorem}
Suppose that $H_1, H_2$ have $m_1, m_2$ edges respectively. Suppose that each edge of $H_i$ intersects with at most $d_i$ other edges of $H_i$, and suppose that
$$
(d_1 + 1) m_2 + (d_2 + 1) m_1 < \frac{\binom{n}{r}}{e}
$$

Then the Swapping Algorithm finds injections $\phi_i: V(H_i) \rightarrow [n]$ such that $\phi_1(H_1)$ is edge-disjoint to $\phi_2(H_2)$.

Suppose further that $r \leq \log^{O(1)} n$ and
$$
(d_1 + 1) m_2 + (d_2 + 1) m_1 < \frac{(1-\epsilon) \binom{n}{r}}{e}
$$

Then the Parallel Swapping Algorithm finds such injections with high probability in $\frac{\log^{O(1)} n}{\epsilon}$ time and using $\text{poly}(m_1, m_2, n)$ processors.
\end{theorem}
\begin{proof}
\cite{random-inj} proves this fact using the LLLL, and the proof immediately applies to the Swapping Algorithm as well. We review the proof briefly: we may assume without loss of generality that the vertex set of $H_1$ is $[n]$ and the vertex set of $H_2$ has cardinality $n$ and that $\phi_1$ is the identity permutation; then we only need to select the bijection $\phi_2: H_2 \rightarrow [n]$. For each pair of edges $e_1 = \{u_1, \dots, u_r \} \in H_1, e_2 = \{v_1, \dots, v_r \} \in H_2$, and each ordering $\sigma \in S_r$, there is a separate bad-event $\phi_2( v_1 ) = u_{\sigma 1} \wedge \dots \wedge \phi_2 (v_r) = u_{\sigma r}$. Now observe that the LLL criterion is satisfied for these bad-events, under the stated hypothesis.

The proof for the Parallel Swapping Algorithm is almost immediate. There is one slight complication: the total number of atomic bad-events is $m_1 m_2 r!$, which could be super-polynomial for $r = \Theta(\log n)$. However, it is easy to see that the total number of bad-events \emph{which are true at any one time} is at most $m_1 m_2$; namely, for each pair of edges $e_1, e_2$, there may be at most one $\sigma$ such that $\phi_2( v_1 ) = u_{\sigma 1} \wedge \dots \wedge \phi_2 (v_r) = u_{\sigma r}$. It is not hard to see that Theorem~\ref{parallel-thm} still holds under this condition.
\end{proof}

\section{Conclusion}
\label{sec:conclusion}
The original formulation of the LLLL \cite{erdos-spencer} applies in a natural way to general probability spaces. There has been great progress over the last few years in developing constructive algorithms, which find in polynomial time the combinatorial structures in these probability spaces whose existence is guaranteed by the LLL. These algorithms have been developed in great generality, encompassing the Swapping Algorithm as a special case. 

However, the Moser-Tardos algorithm has uses beyond simply finding a object which avoids the bad-events. In many ways, the Moser-Tardos algorithm is more powerful than the LLL. We have already seen problems that feature its extensions: e.g., Theorem~\ref{szabo-thm} requires the use of the Partial Resampling variant of the Moser-Tardos algorithm, and Proposition~\ref{it-prop} requires the use of the Moser-Tardos distribution (albeit in the context of the original Moser-Tardos algorithm, not the Swapping Algorithm). 

While the algorithmic frameworks of Achlioptas \& Iliopoulous and Harvey \& Vondr\'{a}k achieve the main goal of a generalized constructive LLL algorithm, they do not match the full power of the Moser-Tardos algorithm. However, our analysis shows that the Swapping Algorithm matches nearly all of the additional features of the Moser-Tardos algorithm. In our view, one main goal of our paper is to serve as a roadmap to the construction of a \emph{true} generalized LLL algorithm. Behind all the difficult technical analysis, there is the underlying theme: even complicated probability spaces such as permutations can be reduced to ``variables'' (the domain and range elements of the range) which interact in a somewhat ``independent'' fashion. 

Encouragingly, there has been progress toward this goal. For example, one main motivation of \cite{achlioptas,harvey} was to generalize the Swapping Algorithm. Then, Kolmogorov noticed in \cite{kolmogorov} that our Swapping Algorithm had a certain nice property, namely the ability to select the resampled bad-event in an arbitrary fashion, that the analysis of \cite{achlioptas} lacked; this led to the work of \cite{kolmogorov} which partially generalized that property (which Kolmogorov refers to as \emph{commutativity}).

At the current time, we do not even know how to define a truly generalized LLL algorithm, let alone analyze it. But we hope that we have at least provided an example approach toward such an algorithm.

\section{Acknowledgments}
We would like to thank the anonymous reviewers of the conference and journal versions of this paper, for their helpful comments and suggestions.

\appendix

\section{Symmetry properties of the swapping subroutine}
\label{symmetry-sec}
In the following series of propositions, we show a variety of symmetry properties of the swapping subroutine. This analysis will use simple results and notations of group theory. We let $S_l$ denote the symmetric group on $l$ letters, which we identify with the set of permutations of $[l]$. We let $(a \ b)$ denote the permutation (of whatever dimension is appropriate) that swaps $a/b$ and is the identity otherwise. We write multiplications on the right, so that $\sigma \tau$ denotes the permutation which maps $x$ to $\sigma(\tau(x))$. Finally, we will sometimes write $\sigma x$ instead of the more cumbersome $\sigma(x)$.

\begin{proposition}
\label{swap-invariant-prop}
The swapping subroutine is invariant under permutations of the domain or range, namely that for any permutations $\tau, \sigma$ we have
$$
P(\text{Swap}(\pi; x_1, \dots, x_r) = \sigma) = P(\text{Swap}(\pi \tau; \tau^{-1} x_1 , \dots, \tau^{-1} x_r) = \sigma \tau)
$$
and
$$
P(\text{Swap}( \pi; x_1, \dots, x_r) =  \sigma) = P(\text{Swap}(\tau \pi; x_1, \dots, x_r = \tau \sigma)
$$
\end{proposition}
\begin{proof}
We prove this by induction on $r$. The following equivalence will be useful. We can view a single call to Swap as follows: we select a random $x'_1$ and swap $x_1$ with $x'_1$; let $\pi' = \pi \cdot (x_1 \ x'_1) $ denote the permutation after this swap. Now consider the permutation on $n-1$ letters obtained by removing $x_1$ from the range and $\pi'(x_1)$ from the range of $\pi'$; we use the notation $\pi' - (x_1, *)$ to denote this restriction of range/domain. We then recursively call $\text{Swap}(\pi' - (x_1, *), x_2, \dots, x_r)$. 

Now, in order to have $\text{Swap}( \pi \tau; \tau^{-1} x_1, \dots, \tau^{-1} x_r) = \sigma \tau$ we must first swap $\tau^{-1} x_1$ with $x'_1 = \tau^{-1} \pi^{-1} \sigma \tau x_1 $; this occurs with probability $1/n$. Then we would have
\begin{align*}
&P(\text{Swap}( \pi \tau; \tau^{-1} x_1, \dots, \tau^{-1} x_r) =  \sigma \tau)  \\
&\qquad = \tfrac{1}{n} P(\text{Swap}( \pi \tau (\tau^{-1} x_1 \ \ \tau^{-1} \pi^{-1} \sigma x_1 )  - (\tau^{-1} x_1, *); \tau^{-1} x_2, \dots,  \tau^{-1} x_r) = \sigma \tau - (\tau^{-1} x_1, *) ) \\
&\qquad = \tfrac{1}{n} P(\text{Swap}( \pi \tau (\tau^{-1} x_1 \ \ \tau^{-1} \pi^{-1} \sigma x_1 ) \tau^{-1} - (x_1, *); \tau^{-1} \tau x_2, \dots,  \tau^{-1} \tau x_r) = \sigma \tau \tau^{-1} - (x_1, *) ) \\
&\qquad \qquad \text{by inductive hypothesis} \\
&\qquad = \tfrac{1}{n} P(\text{Swap}( \pi (x_1 \ \pi^{-1} \sigma x_1 ) \tau^{-1} - (x_1, *); x_2, \dots,  x_r) = \sigma - (x_1, *) ) \\
&\qquad =  P(\text{Swap}( \pi; x_1, x_2, \dots,  x_r) = \sigma)
\end{align*}

A similar argument applies for permutation of the range (i.e., post-composition by $\tau$).
\end{proof}

Also, the order in which we perform the swaps is irrelevant:
\begin{proposition}
\label{swap-invariant-prop2}
Let $\pi \in S_n$ be fixed, and let $x_1, \dots, x_r \in [n]$ be fixed as well. Let $\rho: [r] \rightarrow [r]$ be a permutation on $r$ letters; then for any permutation $\sigma \in S_n$ we have
$$
P(\text{Swap}(\pi; x_1, \dots, x_r) = \sigma) = P(\text{Swap}(\pi; x_{\rho(1)}, \dots, x_{\rho(r)} = \sigma)
$$
\end{proposition}
\begin{proof}
We will prove this by induction on $r$. We assume $\rho(1) \neq 1$ or else this follows immediately from induction.

We have:
\begin{align*}
&P(\text{Swap}(\pi; x_{\rho(1)}, \dots, x_{\rho(r)}) = \sigma)  \\
& \qquad= \tfrac{1}{n} P(\text{Swap}(\pi (x_{\rho(1)} \ \pi^{-1} \sigma x_{\rho(1)}); x_{\rho(2)}, \dots, x_{\rho(r)}) = \sigma)  \\
& \qquad = \tfrac{1}{n} P(\text{Swap}(\pi (x_{\rho(1)} \ \pi^{-1} \sigma x_{\rho(1)}); x_1, x_2, \dots, x_{\rho(1)-1}, x_{\rho(1)+1}, \dots, x_r)  = \sigma)  \qquad \text{by I.H.} \\
& \qquad = \tfrac{1}{n(n-1)} P(\text{Swap}(\pi (x_{\rho(1)} \ \pi^{-1} \sigma x_{\rho(1)})(x_1 \  (\pi (x_{\rho(1)} \ \pi^{-1} \sigma x_{\rho(1)})^{-1}) \sigma x_1)  x_2, \dots, x_{\rho(1)-1}, x_{\rho(1)+1}, \dots, x_r)  = \sigma)  \\
& \qquad = \tfrac{1}{n(n-1)} P(\text{Swap}(\pi (x_{\rho(1)} \ \pi^{-1} \sigma x_{\rho(1)})(x_1 \  (x_{\rho(1)} \ \pi^{-1} \sigma x_{\rho(1)}) \pi^{-1} \sigma x_1)  x_2, \dots, x_{\rho(1)-1}, x_{\rho(1)+1}, \dots, x_r)  = \sigma) 
\end{align*}

At this point, consider the following simple fact about permutations: for any $a_1, a_2, b_1, b_2 \in [l]$ with $a_1 \neq a_2, b_1 \neq b_2$, we have
$$
(a_2 \ b_2) (a_1 \ (a_2 \  b_2) b_1) = (a_1 \ b_1) (a_2 \  (a_1 \ b_1) b_2)
$$
This fact is simple to prove by case analysis considering which of the letters $a_1, a_2, b_1, b_2$ are equal to each other.

We now apply this fact using $a_1 = x_1, a_2 = x_{\rho(1)}, b_1 =\pi^{-1} \sigma a_1, b_2 =  \pi^{-1} \sigma a_2$; this gives us 
\begin{align*}
&P(\text{Swap}(\pi; x_{\rho(1)}, \dots, x_{\rho(r)}) = \sigma)  \\
& \qquad = \tfrac{1}{n(n-1)} P(\text{Swap}(\pi (x_{1} \ \pi^{-1} \sigma x_1)(x_{\rho(1)} \  (x_{1} \ \pi^{-1} \sigma x_{1}) \pi^{-1} \sigma x_{\rho(1)});  x_2, \dots, x_{\rho(1)-1}, x_{\rho(1)+1}, \dots, x_r)  = \sigma) \\
& \qquad = \tfrac{1}{n} P(\text{Swap}(\pi (x_{1} \ \pi^{-1} \sigma x_1); x_{\rho(1)},  x_2, \dots, x_{\rho(1)-1}, x_{\rho(1)+1}, \dots, x_r)  = \sigma) \\
& \qquad = \tfrac{1}{n} P(\text{Swap}(\pi (x_{1} \ \pi^{-1} \sigma x_1); x_2, \dots, x_r) = \sigma) \qquad \text{by I.H.} \\
& \qquad =P(\text{Swap}(\pi; x_1, \dots, x_r) = \sigma)
\end{align*}
\end{proof}

In our analysis and algorithm, we will seek to maintain the symmetry between the ``domain" and ``range" of the permutation. The swapping subroutine seems to break this symmetry, inasmuch as the swaps are all based on the \emph{domain} of the permutation. However, this symmetry-breaking is only superficial as shown in Proposition~\ref{exchange-sym-prop}.

\begin{proposition}
\label{exchange-sym-prop}
Define the alternate swapping subroutine, which we denote $\text{Swap2}(\pi; y_1, \dots, y_r)$ as follows: 
\begin{enumerate}
\item Suppose $\pi$ is a permutation of $[n]$. Repeat the following for $i = 1, \dots, r$:
\item Select $y'_i$ uniformly at random among $[n] - \{y_1, \dots, y_{i-1} \}$. 
\item Swap entries $\pi^{-1}(y_i)$ and $\pi^{-1} (y'_i)$ of the permutation $\pi$.
\end{enumerate}

More compactly:
$$
\text{Swap2}(\pi; y_1, \dots, y_r) = \text{Swap}(\pi^{-1}, y_1, \dots, y_r)^{-1}
$$

Then the algorithms Swap and Swap2 induce the same distribution, namely that if $\pi(x_1) = y_1, \dots, \pi(x_r) = y_r$, then for any permutation $\sigma$ we have
$$
P(\text{Swap}(\pi; x_1, \dots, x_r) = \sigma) = P(\text{Swap2}(\pi; y_1, \dots, y_r) = \sigma)
$$
\end{proposition}
\begin{proof}
A similar recursive definition applies to Swap2 as for Swap: we select $x'_1$ uniformly at random, swap $x_1/x'_1$, and then call $\text{Swap2}(\pi(x_1 \  x_1') - (*, y_1); y_2, \dots, y_r)$. The main difference is that we remove the image point $(*,y_1)$ instead of the domain point $(x_1, *)$.

Now, in order to have $\text{Swap2}(\pi; y_1, \dots, y_r) = \sigma$ we must first swap $x_1$ with $x'_1 = \pi^{-1} \sigma x_1 $; this occurs with probability $1/n$. Next, we recursively call Swap2 on the permutation $\pi (x_1 x'_1) - (*, y_1)$ yielding:
\begin{align*}
P(\text{Swap2}(\pi; y_1, \dots, y_r) = \sigma) &= \tfrac{1}{n} P(\text{Swap2}(\pi (x_1 x'_1) - (*, y_1); y_2, \dots, y_r) = \sigma - (*, y_1) ) \\
&= \tfrac{1}{n} P(\text{Swap}(\pi (x_1 \  x'_1) - (*, y_1); (x_1  \ x'_1) \pi^{-1} y_2, \dots,(x_1  \ x'_1) \pi^{-1} y_r) = \sigma - (*, y_1) ) \\
& \qquad \text{by inductive hypothesis} \\
&= \tfrac{1}{n} P(\text{Swap}(\pi - (x_1, y_1);  x_2, \dots, x_r) = \sigma (x_1 \ x_1')- (x_1, y_1) ) \\
& \qquad \text{by Proposition~\ref{swap-invariant-prop}, when we pre-compose with $(x_1 \ x_1')$} \\
&= \tfrac{1}{n} P(\text{Swap}((\sigma x_1 \ \sigma x'_1) \pi - (x_1, *);  x_2, \dots, x_r) = (\sigma x_1 \ \sigma x'_1) \sigma (x_1 \ x_1')- (x_1,* ) \\
& \qquad \text{by Proposition~\ref{swap-invariant-prop}; when we post-compose with $(\sigma x_1 \ \sigma x'_1)$} \\
&= \tfrac{1}{n} P(\text{Swap}(\pi (x_1 \ x'_1) - (x_1, *);  x_2, \dots, x_r) = \sigma- (x_1, *) \\
&= P(\text{Swap}(\pi; x_1, \dots, x_r) = \sigma)
\end{align*}
\end{proof}

\bibliographystyle{tocplain}   




\end{document}